\documentclass[reqno,10pt]{amsart}
\usepackage{etex}
\usepackage{amsmath}
\usepackage{amssymb}
\usepackage{amsfonts}
\usepackage{dsfont}
\usepackage{amsthm}
\usepackage{amsxtra}
\usepackage{amscd}
\usepackage{portland}
\usepackage{rotating}
\usepackage{nicefrac}
\usepackage[usenames,dvipsnames]{pstricks}
\usepackage{pstricks-add}
\usepackage{pst-node}
\usepackage{pst-slpe}
\usepackage{xspace}
\usepackage{epsfig}
\usepackage{float}
\usepackage[all,web,arc,poly,dvips]{xy}
\usepackage{epic,eepic}
\usepackage{epsfig}
\usepackage{color}
\usepackage{array}
\usepackage{pifont}
\usepackage{multirow}
\usepackage{latexsym}
\usepackage{graphics}
\usepackage{graphicx}
\usepackage{enumerate}
\usepackage[mathscr]{euscript}
\usepackage[fleqn,tbtags]{mathtools}
\usepackage[
        colorlinks,
        hyperfootnotes=false,
        pagebackref,
        hyperindex,
        filecolor=purple,
        urlcolor=purple,
        citecolor=red,
        linkcolor=blue]{hyperref}

\theoremstyle{plain}
\newtheorem{theorem}{Theorem}[section]
\newtheorem{corollary}{Corollary}[section]
\newtheorem{lemma}{Lemma}[section]
\newtheorem{proposition}{Proposition}[section]
\theoremstyle{definition}

\theoremstyle{remark}
\newtheorem{remark}{Remark}[section]

\numberwithin{equation}{section}

\newcommand{\C}{\mathbb{C}}
\newcommand{\R}{\mathbb{R}}
\newcommand{\Z}{\mathbb{Z}}
\newcommand{\N}{\mathbb{N}}
\newcommand{\D}{\mathbb{D}}
\newcommand{\compD}{\bar{\mathbb{D}}}
\newcommand{\T}{\mathbb{T}}

\newcommand{\half}{
        {\lower0.00ex\hbox{\raise.6ex\hbox{\the\scriptfont0 1}
                           \kern-.5em\slash\kern-.1em\lower.45ex
                                     \hbox{\the\scriptfont0 2}}}}
\newcommand{\quarter}{
        {\lower0.00ex\hbox{\raise.6ex\hbox{\the\scriptfont0 1}
                           \kern-.5em\slash\kern-.1em\lower.45ex
                                     \hbox{\the\scriptfont0 4}}}}
\newcommand{\tquarter}{
        {\lower0.00ex\hbox{\raise.6ex\hbox{\the\scriptfont0 3}
                           \kern-.5em\slash\kern-.1em\lower.45ex
                                     \hbox{\the\scriptfont0 4}}}}
\newcommand{\eighth}{
        {\lower0.00ex\hbox{\raise.6ex\hbox{\the\scriptfont0 1}
                           \kern-.5em\slash\kern-.1em\lower.45ex
                                     \hbox{\the\scriptfont0 8}}}}
\newcommand{\othird}{
        {\lower0.00ex\hbox{\raise.6ex\hbox{\the\scriptfont0 1}
                           \kern-.5em\slash\kern-.1em\lower.45ex
                                     \hbox{\the\scriptfont0 3}}}}

\def\arg{{\rm arg}}

\begin{document}

\title[]{Loop Equation Analysis of the Circular $ \beta $ Ensembles}

\author{N.S.~Witte and P.J.~Forrester}
\address{Department of Mathematics and Statistics,
University of Melbourne,Victoria 3010, Australia}
\email{\tt nsw@ms.unimelb.edu.au, pjforr@unimelb.edu.au}

\begin{abstract}
We construct a hierarchy of loop equations for invariant circular ensembles. These are valid for general classes
of potentials and for arbitrary inverse temperatures $ {\rm Re}\,\beta>0 $ and number of eigenvalues $ N $. 
Using matching arguments for the resolvent functions of linear statistics $ f(\zeta)=(\zeta+z)/(\zeta-z) $
in a particular asymptotic regime, the global regime, we systematically develop the corresponding large
$ N $ expansion and apply this solution scheme to the Dyson circular ensemble.
Currently we can compute the second resolvent function to ten orders in this expansion and also its general
Fourier coefficient or moment $ m_{k} $ to an equivalent length. The leading large $ N $, large $ k $, $ k/N $
fixed form of the moments can be related to the small wave-number expansion of the structure function in the
bulk, scaled Dyson circular ensemble, known from earlier work. From the moment expansion we conjecture some
exact partial fraction forms for the low $ k $ moments. For all of the forgoing results we have made a
comparison with the exactly soluble cases of $ \beta = 1,2,4 $, general $ N $ and even, positive $ \beta $, 
$ N=2,3 $.
\end{abstract}

\subjclass[2010]{15B52, 28C10, 47A10, 60J65, 42A05}
\maketitle

\section{Introduction}\label{Intro}
\setcounter{equation}{0}

Fundamental to random matrix theory is a set of equations know variously as Virasoro constraints,
Ward identities, Schwinger-Dyson equation, Pastur equations or loop equations. We will use the latter
terminology. These allow, in principle at least, the computation of the large $N$ global scaled asymptotic
expansion of correlation functions for eigenvalue probability density functions (PDFs) of the form
\begin{equation}
  e^{-\sum_{j=1}^N V(x_j)} \prod_{1 \le j < k \le N} |x_k - x_j|^\beta, \qquad -\infty < x_j < \infty \: \: (j=1,\ldots,N) .
\label{1.1}
\end{equation}
Here $V(x)$ is referred to as the potential (for Gaussian ensembles $V(x)$ is proportional to $x^2$),
while $\beta =2\kappa > 0$ is sometimes called the Dyson index, with $\beta = 1,2,4$ corresponding to matrices
with orthogonal, unitary and symplectic symmetry respectively (see e.g.~\cite[Ch.~1]{For_2010}).

We recall that global scaling refers to a rescaling of the eigenvalues so that their support is a single finite interval
(single cut), or a collection of finite intervals (multiple cuts). As a concrete example, consider the Gaussian
orthogonal ensemble of real symmetric matrices, defined as the set of matrices of the form $G=(X + X^T)/2$,
where $X$ is an $N \times N$ matrix with entries independent standard Gaussians. The eigenvalue PDF is
given by (\ref{1.1}) with $V(x) = x^2/2$ and $\beta = 1$ (see e.g.~\cite[Prop.~1.3.4]{For_2010}). By rescaling
$\lambda_j \mapsto \sqrt{2N} \lambda_j$, the leading order support of the spectral density $\rho_{(1)}(\lambda;N)$ is
the interval $(-1,1)$.

Let $\rho_{(1)}(\lambda;N)$ denote the one-point function (eigenvalue density) with global scaling, normalised to 
integrate to unity.
The loop equations allow the computation of the large $N$ asymptotic expansion of the resolvent
\begin{equation}
  R(x;N) := \int_{-\infty}^\infty \frac{\rho_{(1)}(\lambda;N)}{x - \lambda} \, d \lambda = R_0(x) + \frac{1}{N} R_1(x) + \ldots , 
\label{R1}
\end{equation}
where
\begin{equation*}
  R_0(x) = 2 \left[ x - \sqrt{x^2-1} \right], \qquad R_1(x) = \Big ( \frac{1}{\beta} - \frac{1}{2} \Big )\left[ \frac{1}{\sqrt{x^2-1}} - \frac{x}{x^2-1} \right] ,
\end{equation*}
up to and including terms $ {\rm O}(N^{-6}) $ \cite{WF_2014}, \cite{MMPS_2012}, \cite{BMS_2011}.
The  asymptotic expansion of the (smoothed) eigenvalue density follows from the inverse Cauchy transform
\begin{equation*}
  \rho_{(1)}(\lambda;N) = \frac{1}{2\pi i} \lim_{\epsilon \to 0^+}\Big( R(\lambda - i\epsilon) - R(\lambda + i\epsilon) \Big),
\end{equation*}
and gives, with $\chi_{\lambda \in J} = 1$ for $\lambda \in J$ and $\chi_{\lambda \in J} = 0$ otherwise,
\begin{multline}
  \rho_{(1)}(\lambda;N) = \frac{1}{\pi} \sqrt{1 - \lambda^2} \chi_{\lambda \in (-1,1)}
\\
  + \frac{1}{N} \Big ( \frac{1}{\beta} - \frac{1}{2} \Big )\left[ \frac{1}{2}(\delta(\lambda - 1) + \delta(\lambda + 1) )-
    \frac{1}{\pi \sqrt{1 - \lambda^2}} \chi_{\lambda \in (-1,1)} \right] + {\rm O} \Big ( \frac{1}{N^2} \Big ).
\label{R1a}
\end{multline}
Here the leading term is the celebrated Wigner semi-circle law.

Our interest in this paper is in the loop equation formalism for generalised circular ensembles. The latter is the
class of eigenvalue PDFs that extend (\ref{1.1}) from the real line to unit circle, and are thus of the form
\begin{equation}
  e^{-\sum_{j=1}^N V(\theta_j)} \prod_{1 \le j < k \le N} |e^{i \theta_k}-e^{i \theta_j}|^\beta, \qquad 0 \leq \theta_j < 2 \pi.
\label{1.1a}
\end{equation}
We are motivated by some earlier work of one of the present authors and collaborators
\cite{FJM_2001}. That work relates to the bulk scaled limit of the two-point correlation function for the circular
ensemble (\ref{1.1a}) with $V(\theta)$ independent of $\theta$. This was first isolated by Dyson \cite{Dys_1962d} in the
study of unitary analogues of the Gaussian ensembles. Denoting the PDF by
$p_N(\theta_1,\ldots,\theta_N)$, the two-point correlation function $\rho_{(2)}$ is specified by
\begin{align}
  \rho_{(2)}(\theta_1,\theta_2;N) & =\rho_{(2)}(\theta_2 - \theta_1,0;N)  \nonumber \\
  & = N(N-1) \int_0^{2 \pi} d \theta_3 \cdots \int_0^{2 \pi} d \theta_N \,p_N(\theta_1,\theta_2,\theta_3,\ldots, \theta_N),
\label{R1b}
\end{align}
and its bulk scaling limit by
\begin{equation*}
  \rho_{(2)}^{\rm bulk}(s,0) = \lim_{N \to \infty} (2 \pi/N)^2 \rho_{(2)}(0,2 \pi s/N;N).
\end{equation*}
In terms of $\rho_{(2)}^{\rm bulk}$ one defines the structure function
\begin{equation}
  S(k;\beta) := \int_{-\infty}^\infty \Big ( \rho_{(2)}^{\rm bulk}(s,0) - 1 \Big ) e^{i k s} \, ds, \qquad (k \ne 0).
\label{S}
\end{equation}
With $\kappa = \beta/2$, $y = |k|/\pi \beta$, one of the main results of \cite{FJM_2001} is the expansion
\begin{multline}
  \frac{\pi \beta}{|k|} S(k;\beta) = 1 + (\kappa - 1) y + (\kappa - 1)^2 y^2 + (\kappa - 1) \left( \kappa^2 - \frac{11}{6} \kappa + 1 \right)y^3
\\
  + (\kappa-1)^2\left( \kappa^2-\frac{3}{2}\kappa+1 \right)y^4 + (\kappa-1)\left( \kappa^4-\frac{91}{30}\kappa^3+\frac{62}{15}\kappa^2-\frac{91}{30}\kappa+1 \right)y^5 + \cdots,
\label{Sa}
\end{multline}
up to and including the term $O(y^9)$. 

To deduce (\ref{Sa}), it was assumed that the small $k$ expansion
of the structure function is of the form
\begin{equation}
  \frac{\pi \beta}{|k|} S(k;\beta)  = 1 + \sum_{j=1}^\infty p_j(\kappa) y^j,
\label{S1}
\end{equation}
where $p_j(x)$ is a polynomial of degree $j$.
Moreover, with
\begin{equation}
  f(k;\beta) := \frac{\pi \beta}{|k|} S(k;\beta), \: \: \: 0 < k < {\rm min} \,(2 \pi,\pi \beta),
\label{S2}
\end{equation}
and $f$ defined by analytic continuation for $k < 0$, it is a rigorous result that \cite{FJM_2001}
\begin{equation}
  f(k;\beta)  = f \Big ( - \frac{2k}{\beta}; \frac{4}{\beta} \Big ) .
\label{S3}
\end{equation}
This applied to (\ref{S1}) requires that the polynomials in (\ref{S1}) have the reciprocal property
\begin{equation}
  p_j(1/x) = (-1)^j x^{-j} p_j(x) .
\label{S4}
\end{equation}
Exact results for $\beta \to 0$, $\beta = 2$, $\beta = 4$ and first order expansions about $\beta = 2$ and $\beta = 4$ were
then used to determine the independent coefficients in the polynomial up to the highest order possible. We will show in the present
paper that the loop equations provide a systematic approach to the generation of the expansion (\ref{Sa}).

Our key results consist of two parts - a full and complete constructive proof of the hierarchy of loop equations for 
circular $\beta$ ensembles in Propositions \ref{LoopEqn} and \ref{LoopEqn2}, and the application of this system of
loop equations to the Dyson circular $\beta$ ensemble upon specialisation of the forgoing theory. We have not seen
this hierarchy written down in the literature and while it has resemblances with the system of loop equations on $ \R $
it differs in many significant details. This resemblance is taken up in the discussion contained in \S \ref{Survey}. 

For the Dyson circular $\beta$ ensembles we give exact results for the moments $ m_k $ appearing as the Fourier
coefficients of the connected two-point correlation function or density (here $ \theta=\theta_2-\theta_1 $)
\begin{equation}
   \rho_{(2)C}(\theta_1,\theta_2;N) \equiv \rho_{(2)}(\theta_1,\theta_2;N) - \rho_{(1)}(\theta_1;N)\rho_{(1)}(\theta_2;N)
   = \sum_{k\in \mathbb{Z}} m_{k}e^{ik\theta} ,
\label{2ptConnected}
\end{equation}
in terms of rational partial fractions for low index $ k $ (see Prop. \ref{mRational})
\begin{gather*}
   m_{0}(N,\kappa) = -N ,
\notag \\
   m_{1}(N,\kappa) = -N+\frac{1}{\kappa}+\frac{(\kappa-1)}{\kappa(\kappa N+1-\kappa)} ,
\notag \\
   m_{2}(N,\kappa) = -N+\frac{2}{\kappa}
\notag \\
        +\frac{(\kappa-1)}{\kappa}\left[ \frac{2}{\kappa N+1-\kappa}-\frac{2(\kappa-2)}{(\kappa+1)(\kappa N+2-\kappa)}+\frac{2(2\kappa-1)}{(\kappa+1)(\kappa N+1-2\kappa)} \right] .
\end{gather*}
In addition we give the large $N$ expansion of the $ m_k $ for fixed but arbitrary $ k < {\rm O}(N) $ in a particular regime,
which we call the {\it global regime}, in two ways - a direct one relating to exact forms above (see Corr. \ref{momentLargeN})
and another through the
generating function, the two-point connected resolvent function, or essentially the Riesz-Herglotz transform of the above two-point density 
\begin{equation*}
   W_{2}(z_{1},z_{2}) = W_{2}(z=z_{2}/z_{1}) = -m_{0}-N - 4\sum_{k=1}^{\infty}(m_{k}+N)z^{k} ,
\end{equation*}
where the leading terms are (see Prop. \ref{W2expand})
\begin{multline*}
    W_{2}(z_{1},z_{2}) = -\frac{4}{\kappa}\frac{z_1z_2}{(z_{1}-z_{2})^2}
    -4\frac{(\kappa-1)}{\kappa^2N}\frac{(z_1+z_2)z_1z_2}{(z_{1}-z_{2})^3}
\\
    -4\frac{(\kappa -1)^2}{\kappa^3 N^2}\frac{z_1z_2}{\left(z_1-z_2\right)^4}\left[ (z_1+z_2)^2+2 z_1z_2 \right] 
\\
    -4\frac{(\kappa-1)}{\kappa^4 N^3}\frac{(z_1+z_2) z_1z_2}{\left(z_1-z_2\right)^5}\left[ (\kappa -1)^2 (z_1+z_2)^2+2 \left(4 \kappa ^2-7 \kappa +4\right) z_1z_2 \right]
    +\ldots .
\end{multline*}
It is interesting to note the appearance of the {\it Koebe function} in our setting as the leading 
order and universal coefficient in $ W_2 $, (see \eqref{koebe}). This function occupies an important role in
the theory of univalent functions, \cite{Dur_1983}, \cite{Gri_1999}, \cite{Koe_2003},
being the unique extremal example of such functions. However it is not clear how considerations arising
from geometric function theory have interpretations in the context of the Dyson circular ensembles.
We observe that the analytic properties of the circular ensembles differ markedly from Hermitian ensembles,
in that the resolvent functions possess convergent expansions and not formal ones. This is related to the fact that
under stereographic projection $ e^{i\theta}=\frac{\displaystyle 1+ix}{\displaystyle 1-ix} $ the Dyson circular
ensemble is equivalent to the Cauchy $ \beta $ ensemble with weight
\begin{equation*}
  w(x) = \frac{1}{(1+x^2)^{\kappa(N-1)+1}}, \quad x\in \R ,
\end{equation*}
i.e. the potential has logarithmic growth and is not in the same universality class as say the Gaussian $\beta$
ensembles.

As seen in the case of Hermitian matrices revised in the 2nd and 3rd paragraphs,
and in the summary of some of the results to be derived for the circular ensembles, the loop equation
analysis of correlation functions applies to the global scaling regime. In this regime, the length scales
are effectively macroscopic. Using different methods of analysis, typically based on Jack polynomial
theory (see \cite{For_2010}, Ch.~12), correlations in local regimes on the length scales of the inter-eigenvalue
spacings can be probed. References on that topic include \cite{For_2010}, Ch.~13 and \cite{DL_2011,DL_2014,Liu_2014}.

The plan of our work is as follows: In \S \ref{CircularDefn} we define the fundamental resolvent 
functions required in the theory and give some of their analytic and symmetry properties. The hierarchy
of loop equations is derived in \S \ref{CircularLoopEqn} for a general class of potentials. A solution 
scheme to the loop equations is proposed for one of the large $ N $ regimes, based upon matching arguments in
the decay of the resolvent functions, in \S \ref{CircularSolnScheme}
and a solution scheme specialised to the Dyson ensemble is outlined. This is where our main results of the 
computer algebra calculations are given. As a reference point to the previous sections we augment the well-known 
results for the two point correlations for $ \beta=1,2,4 $, general $N$ in \S 
\ref{beta124DysonEnsembles}-\ref{COE} and
for any even, positive $ \beta $ and low values of $ N=2,3 $ in \S \ref{CorrelationDuality}, and discuss
the comparison of these special cases with those of general $ \kappa $.
In the final section of the present paper, \S \ref{Survey}, we review earlier work on loop equations for
circular ensembles so as to both contrast our contribution, and to put it in context.

\section{Definitions for the General \texorpdfstring{$ N $}{N}, \texorpdfstring{$ \beta $}{b} Circular Ensembles}\label{CircularDefn}
\setcounter{equation}{0}

The unit circle is denoted $ \mathbb{T}=\{z\in \mathbb{C}: |z|=1 \} $, the open unit disc is $ \D=\{z\in \mathbb{C}: |z|<1 \} $
and its exterior is $ \mathbb{\bar{D}}=\{z\in \mathbb{C}: |z|>1 \} $.
The total number of particles in the system is $ N $ including the number of test particles.
On the unit circle the co-ordinates are $ \zeta = e^{i\theta} $, and thus $ |\zeta|=1 $, with arguments $ \theta\in [0,2\pi ] $.
The inverse temperature is $ \beta = 2\kappa $ and usually defined on $ \C\backslash\{0\} $.
Complex co-ordinates $ z, z_1, \ldots $ are generally defined on the Riemann sphere $ \C^{\star} $.
The measure $ d\mu $ is taken to be absolutely continuous on $ \T $ with density $ w $ of the form 
\begin{equation}
  d\mu(\zeta) = e^{-V(\zeta)}\frac{d\zeta}{2\pi i\zeta} = w(\zeta)\frac{d\zeta}{2\pi i\zeta} .
\label{Tmeasure}
\end{equation}
An example of the class of potentials that can be admitted are those drawn from the class of Laurent polynomials $ \mathbb{C}[\zeta,\zeta^{-1}] $ 
with the structure
\begin{equation}
   V(\zeta) = \sum^{M_{+}}_{m\geq 1}t_{m}\zeta^m + \sum^{M_{-}}_{m\geq 1}t_{-m}\zeta^{-m} .
\label{potential}
\end{equation}
A vast literature studying the simplest case of the above example, $ M_+=M_-=1 $, in the context of unitary matrix
models was initiated in the works \cite{GW_1980}, \cite{BG_1980}, which were known to arise as a one-plaquette lattice
model of 2-D Yang-Mills theory.

However, and we wish to emphasis this point, that we admit potentials with a finite number of isolated singularities at 
$ z_{s} \in \D $ or $ z_{s} \in \compD $, and even on $ \T $ however subject to additional restrictions. Due to the
homotopical inequivalence of closed loops on the punctured Riemann sphere to those on the unpunctured sphere it will not
be permissible in general to contract the integration contour $ \T $ to an interval of the real line.
Secondly, even when the forgoing contraction is permitted, unless there is additional symmetry
(e.g. evenness with respect to $ \theta = \arg(\zeta) $) the projection of the lower and upper arcs onto the interval
$ \mathscr{I} $ will lead to two, albeitly related, distinct weights $ w(x) $, $ x\in \mathscr{I} $. 
Further insight into this issue will be provided in the discussion contained in \S \ref{Survey}.

As a minimum requirement on the potential we will henceforth assume the existence of all trigonometric moments of the form
\begin{equation}
\begin{split}
   \int_{\mathbb{T}}\frac{d\zeta}{2\pi i \zeta}e^{-V(\zeta)} \zeta^{m} < \infty
   \\
   \int_{\mathbb{T}}\frac{d\zeta}{2\pi i \zeta}e^{-V(\zeta)} \frac{\zeta+z}{\zeta-z}[V'(\zeta)-V'(z)] \zeta^{m} < \infty
\end{split} , \quad m \in\Z, z\in\C^{\star} .
\label{assume2}
\end{equation} 
Furthermore we will generally require the winding number of $ w(\zeta) $ about $ \zeta=0 $ to vanish
\begin{equation}
    \left. e^{-V(\zeta)} \right|^{{\rm arg}(\zeta)=2\pi}_{{\rm arg}(\zeta)=0} = 0,
\label{assume1}
\end{equation}
however even this can be relaxed within our formalism, after the inclusion of additional boundary terms.

Our ensemble is defined simply through the eigenvalue probability density function 
\begin{equation}
  p(\zeta_1,\ldots,\zeta_N) = \frac{1}{Z_{N}} \prod^{N}_{j=1} w(\zeta_{j}) \prod_{1\leq j<k\leq N} |\zeta_{j}-\zeta_{k}|^{2\kappa} ,
\label{CEpdf}
\end{equation}
where the normalisation is specified by
\begin{equation}
  Z_{N} =  \int_{\mathbb{T}}\frac{d\zeta_{1}}{2\pi i \zeta_{1}} \cdots \int_{\mathbb{T}}\frac{d\zeta_{N}}{2\pi i \zeta_{N}}
           \prod^{N}_{j=1} w(\zeta_{j}) \prod_{1\leq j<k\leq N} |\zeta_{j}-\zeta_{k}|^{2\kappa} .
\label{Znorm}
\end{equation}
Averages of linear statistics of the eigenvalues are defined by
\begin{equation}
   \big\langle \sum^{N}_{r=1}f(\zeta_{r}) \big\rangle
  :=  \frac{1}{Z_{N}}\int_{\mathbb{T}}\frac{d\zeta_{1}}{2\pi i \zeta_{1}} \cdots \int_{\mathbb{T}}\frac{d\zeta_{N}}{2\pi i \zeta_{N}}
           \sum^{N}_{r=1}f(\zeta_{r}) \prod^{N}_{j=1} w(\zeta_{j}) \prod_{1\leq j<k\leq N} |\zeta_{j}-\zeta_{k}|^{2\kappa} ,
\label{CEaverage}
\end{equation}
with the implied normalisation $ \langle 1\rangle=1 $. Defining $ \zeta_{1}=e^{i\theta_{1}} $, $ \zeta_{2}=e^{i\theta_{2}} $,
the density and the two-point correlation function are given as
\begin{gather*}
    \rho_{(1)}(\theta_1;N) = \frac{N}{Z_{N}} \int\frac{d\zeta_2}{2\pi i\zeta_2}\cdots\int\frac{d\zeta_N}{2\pi i\zeta_N}
                         \prod_{j=1}^{N}w(\zeta_j) \prod_{1\leq j<k\leq N}|\zeta_j-\zeta_k|^{2\kappa} , 
\\
    \rho_{(2)}(\theta_1,\theta_2;N) = \frac{N(N-1)}{Z_{N}} \int\frac{d\zeta_3}{2\pi i\zeta_3}\cdots\int\frac{d\zeta_N}{2\pi i\zeta_N}
                         \prod_{j=1}^{N}w(\zeta_j) \prod_{1\leq j<k\leq N}|\zeta_j-\zeta_k|^{2\kappa} ,
\end{gather*}
the latter having been introduced in \eqref{R1b}.

Central to our theory are the resolvent functions which will serve as generating functions for the moments of the
eigenvalues by virtue of the interior and exterior geometrical expansions of the Riesz-Herglotz kernel
\begin{equation}
  \frac{\zeta+z}{\zeta-z} =
   \begin{cases}\displaystyle
     1+2\sum^{\infty}_{l=1}\frac{z^l}{\zeta^l} , &  |z|<|\zeta| \\ 
                \displaystyle
    -1-2\sum^{\infty}_{l=1}\frac{\zeta^l}{z^l} , &  |z|>|\zeta| \\                                      
   \end{cases} .
\label{R-Hkernel} \\ 
\end{equation}
In fact averages with this kernel for the linear statistic, while uncommon in applications of the loop equation method,
are not novel in studies of unitary matrix models when one recognises that through $ \zeta = e^{i\theta}, z = e^{i\phi} $
\begin{equation*}
   \frac{\zeta+z}{\zeta-z} = -i \cot \left( \frac{\theta-\phi}{2} \right) ,
\end{equation*} 
(see the remarks associated with Eq. (3.2) of \cite{Miz_2005}). The Riesz-Herglotz kernel, or cotangent kernel,
is particularly adapted to the circular case for another reason - it appears in the saddle point equations
for the eigenvalue probability density functions of the form
\begin{equation*}
    \prod^{N}_{j=1} e^{-\frac{1}{g}V(e^{i\theta_j})} \prod_{1\leq j<k\leq N} \sin^2 \left( \frac{\theta_j-\theta_k}{2} \right) ,
\end{equation*}
(see Eq. (3.1) of \cite{Miz_2005}).

The first of a sequence of resolvent functions, the Carath\'eodory function, is defined by
\begin{multline}
   W_{1}(z) = \Big \langle \sum_{j} \frac{\zeta_{j}+z}{\zeta_{j}-z} \Big \rangle
            = \int_{\mathbb{T}}\frac{d\zeta_{1}}{2\pi i \zeta_{1}}\frac{\zeta_1+z}{\zeta_1-z}\rho_{(1)}(\theta_1)
\\ =     
   \begin{cases}                                                
     \rho_{0}+2\sum^{\infty}_{l=1}\rho_{l}z^l &  z\in \D \\                                      
    -\rho_{0}-2\sum^{\infty}_{l=1}\rho_{-l}z^{-l} &  z\in \compD \\                                      
   \end{cases} ,
\label{1stCorrelator}                                                 
\end{multline}
with Fourier coefficients $ \rho_{l} = \langle \sum_{p} z_{p}^{-l} \rangle, \, l\in \mathbb{Z} $.

Our first definition of a cumulant $ \langle A_{1} \cdots A_{m+1} \rangle_{c} $ for $ m\geq 0 $ is given implicitly
in terms of the average $ \langle A_{1} \cdots \rangle $ as
\begin{equation}
 \langle A_{1} \cdots A_{m+1} \rangle
  = \sum_{k=1}^{m+1} \sum_{I_{1} \cup \cdots \cup I_{k} = \{1,\ldots,m+1 \} }\prod_{j=1}^{k} \langle A_{I_{j}} \rangle_{c} .
\label{moment-cumulant}
\end{equation}
This definition differs from other authors, such as Mehta \cite{Meh_2004} by a factor of a sign, but our
definition conforms to the more usual statistical conventions, see \S 3.12 of \cite{KS_1969} 
or \S 15.10, or pg. 186 of \cite{Cra_1999}. In contrast to \eqref{moment-cumulant} Mehta's definition 
Eq. (5.1.4) has sign factors. For example in \S 5.1.1 of Mehta \cite{Meh_2004} the connected two-point
correlation function is defined as the negative of \eqref{2ptConnected}. The unconnected resolvent function
or moment of the linear statistic \eqref{R-Hkernel} is defined by 
\begin{equation*}
   U_{n}(z_{1},\ldots,z_{n}) := \left \langle \sum_{j_{1}} \frac{\zeta_{j_{1}}+z_{1}}{\zeta_{j_{1}}-z_{1}} \times\cdots\times 
                                             \sum_{j_{n}} \frac{\zeta_{j_{n}}+z_{n}}{\zeta_{j_{n}}-z_{n}} \right \rangle,
                                             \, n \geq 1; \quad U_0:=1 ,
\end{equation*} 
whereas the connected resolvent function or cumulant is defined as
\begin{equation*}
   W_{n}(z_{1},\ldots,z_{n}) := \left \langle \sum_{j_{1}} \frac{\zeta_{j_{1}}+z_{1}}{\zeta_{j_{1}}-z_{1}} \times\cdots\times 
                                             \sum_{j_{n}} \frac{\zeta_{j_{n}}+z_{n}}{\zeta_{j_{n}}-z_{n}} \right \rangle_{c}, \, n \geq 1 .
\end{equation*}
In particular our study will focus on the second cumulant, which through a simple calculation is related to
the first two densities by the integral formula
\begin{multline*}
  W_{2}(z_1,z_2) = \int_{\mathbb{T}}\frac{d\zeta_{1}}{2\pi i \zeta_{1}}\int_{\mathbb{T}}\frac{d\zeta_{2}}{2\pi i \zeta_{2}}
                   \frac{\zeta_1+z_1}{\zeta_1-z_1}\frac{\zeta_2+z_2}{\zeta_2-z_2} \rho_{(2)C}(\theta_1,\theta_2)
\\
                 + \int_{\mathbb{T}}\frac{d\zeta_{1}}{2\pi i \zeta_{1}}
                   \frac{\zeta_1+z_1}{\zeta_1-z_1}\frac{\zeta_1+z_2}{\zeta_1-z_2} \rho_{(1)}(\theta_1) .
\end{multline*}

We also require the potential resolvent functions, which are defined in their unconnected form by
\begin{multline*}
   Q_{n+1}(z;z_{1},\ldots,z_{n})
\\ := \left \langle \sum_{j_{0}} \frac{\zeta_{j_{0}}+z}{\zeta_{j_{0}}-z}\left[ V'(\zeta_{j_{0}})-V'(z) \right]
                                          \times \sum_{j_{1}} \frac{\zeta_{j_{1}}+z_{1}}{\zeta_{j_{1}}-z_{1}}
                              \times\cdots\times \sum_{j_{n}} \frac{\zeta_{j_{n}}+z_{n}}{\zeta_{j_{n}}-z_{n}} \right \rangle, \, n \geq 0 ,
\end{multline*}
and their connected version by
\begin{multline*}
   P_{n+1}(z;z_{1},\ldots,z_{n})
\\ := \left \langle \sum_{j_{0}} \frac{\zeta_{j_{0}}+z}{\zeta_{j_{0}}-z}\left[ V'(\zeta_{j_{0}})-V'(z) \right]
                                          \times \sum_{j_{1}} \frac{\zeta_{j_{1}}+z_{1}}{\zeta_{j_{1}}-z_{1}}
                              \times\cdots\times \sum_{j_{n}} \frac{\zeta_{j_{n}}+z_{n}}{\zeta_{j_{n}}-z_{n}} \right \rangle_{c}, \, n \geq 0 .
\end{multline*}
In addition to the definition \eqref{moment-cumulant} the moments and cumulants are related through their formal exponential
generating functions by an equivalent definition
\begin{equation}
   \sum^{\infty}_{n=0}\frac{t^n}{n!}U_{n} = \exp\left( \sum^{\infty}_{n=1}\frac{t^n}{n!}W_{n} \right) ,
\label{M-C_exponential}
\end{equation} 
and the related recursive relation
\begin{equation}
   U_{l+1} = \sum_{m=0}^{l} \binom{l}{m}U_{m}W_{l+1-m} .
\label{M-C_recursive}
\end{equation} 
However we will require a more refined recursive relation which properly recognises the arguments of the 
resolvents. In addition we will generally not assume symmetry in the arguments and therefore preserve their
order, so that when combining sets of these we will perform a string concatenation operation, denoted $ \| $,
rather than the set union. Also $ I\backslash I_j $ will denote the excision of the variables in $ I_j $ from
those of $ I $ whilst retaining the original order. We state these generalised results without proof (these
follow from the $ r_1= \cdots = r_{m+1}=1 $ case of Eq. (10) of \cite{Smi_1995}).
\begin{theorem}[\cite{Smi_1995}]\label{MCrelation}
Let $ I=(z_1,\ldots,z_l) $ and we designate $ z_{l+1} $ to be a distinguished variable. The moments $ U_l $
and the cumulants $ W_l $ satisfy the recursive relation, which is a generalisation of \eqref{M-C_recursive}
\begin{equation}
   U_{l+1}(I\| z_{l+1}) = \sum_{I_{j} \subseteq I} W_{l+1-\#(I_{j})}(I\backslash I_{j}\| z_{l+1})U_{\#(I_{j})}(I_{j}) .
\label{U-W_recur}
\end{equation} 
The analogous result for the potential resolvents $ Q_l $ and $ P_l $ is the following recursive relation,
where $ z $ is the distinguished variable
\begin{equation}
   Q_{l+1}(z;I) = \sum_{I_{j} \subseteq I} P_{l+1-\#(I_{j})}(z;I\backslash I_{j})U_{\#(I_{j})}(I_{j}) ,  
\label{Q-P_recur}
\end{equation}
and whereby convention $ U_{0}(\emptyset)=1 $.
\end{theorem}

The moments $ U_n $ and therefore the cumulants $ W_n $ are sectionally analytic with respect to $ z_1, \ldots, z_n $
if the variables are strictly $ z_{j}\in\D $ or
$ z_{j}\in\compD $ as one can see from simple bounds on the remainder terms for $ m\in \N $
\begin{align*}
   \left| \frac{\zeta+z}{\zeta-z}-1-2\sum^{m}_{l=1}\frac{z^l}{\zeta^l} \right|
  & \leq 2\frac{|z|^{m+1}}{1-|z|}, \quad z\in\D ,
\\
   \left| \frac{\zeta+z}{\zeta-z}+1+2\sum^{m}_{l=1}\frac{\zeta^l}{z^l} \right|
  & \leq 2\frac{1}{|z|^{m}(|z|-1)}, \quad z\in\compD .
\end{align*}
Thus there are at most $ 2^n $ distinct functions for each $ W_{n} $ labelled by the string $ D=(d_1,\ldots,d_n) $
with $ d_j \in\{0,\infty\} $.

There are a number of trivial identities and properties satisfied by the cumulants (and moments) which we list for
subsequent use -\\
\begin{enumerate}[(i)]
\item 
re-labelling symmetry $ 1\leq i \leq n $ and for all $ \sigma\in S_n $
\begin{equation}
   W_{n}\left(\begin{array}{ccc} \ldots, & d_{\sigma(i)}, & \ldots \\ \ldots, & z_{\sigma(i)}, & \ldots \end{array}\right)
 = W_{n}\left(\begin{array}{ccc} \ldots, & d_{i}, & \ldots \\ \ldots, & z_{i}, & \ldots \end{array}\right) ;
\label{SYM-labels}
\end{equation}
\item
permutation symmetry within the subsets of variables in $ \infty $ and $ 0 $ domains respectively
\begin{equation}
   W_{n}\left(\begin{array}{ccc} \infty^{\#(D_{\infty})} & , & 0^{n-\#(D_{\infty})} \\ \sigma(Z_{\infty}) & , & \sigma'(Z_{0}) \end{array}\right)
 = W_{n}\left(\begin{array}{ccc} \infty^{\#(D_{\infty})} & , & 0^{n-\#(D_{\infty})} \\ Z_{\infty} & , & Z_{0} \end{array}\right) ,
\label{SYM-permute}
\end{equation}
for $ \sigma\in S_{\#(D_{\infty})} $, $ \sigma'\in S_{n-\#(D_{\infty})} $, $ D_{\infty}\| D_{0} = D $,
$ Z_{\infty}=(\ldots,z_j,\ldots) $ such that $ d_j=\infty $, and $ Z_{0}=(\ldots,z_j,\ldots) $ such that $ d_j=0 $.
Properties (i) and (ii) imply that one can re-order the domains and variables so that
$ d_1 = \ldots = d_{\#(D_{\infty})} = \infty $ and $ d_{\#(D_{\infty})+1} = \ldots = d_n = 0 $;
\item
reduction in index $ 1 \leq m \leq n$
\begin{equation}
   U_{n}\left(\begin{array}{ccc} \ldots, & d_{m}=\substack{\displaystyle 0 \\\displaystyle \infty}, & \ldots \\
                                 \ldots, & z_{m}=\substack{\displaystyle 0 \\\displaystyle \infty}, & \ldots \end{array}\right)
 = \pm N U_{n-1}\left(\begin{array}{c} \ldots ,d_{m-1}, d_{m+1}, \ldots \\
                                       \ldots ,z_{m-1}, z_{m+1}, \ldots \end{array}\right) ;
\label{SYM-reduce}
\end{equation} 
\item
special values in $ \C^{\star} $, $ 0\leq m\leq n $
\begin{gather*}
   U_{n}\left(\begin{array}{cccccc} 0, & \ldots, & 0, & \infty, & \ldots, & \infty \\ z_1=0, & \ldots, & z_m=0, & z_{m+1}=\infty, & \ldots, & z_{n}=\infty \end{array}\right)
   = (-1)^{n-m}N^n,
\\
   W_{1}\left(\begin{array}{c} 0 \\ z=0 \end{array}\right) = N, \; W_{1}\left(\begin{array}{c} \infty \\ z=\infty \end{array}\right) = -N,
\end{gather*}
for $ n\geq 2 $
\begin{gather}
   W_{n}\left(\begin{array}{ccc} 0, & \ldots, & 0 \\ z_1=0, & \ldots, & z_n=0 \end{array}\right)
 = W_{n}\left(\begin{array}{ccc} \infty, & \ldots, & \infty \\ z_1=\infty, & \ldots, & z_{n}=\infty \end{array}\right)
 = 0 .
\label{SYM-special}
\end{gather} 
\end{enumerate}

\section{Loop Equations for general \texorpdfstring{$ N $}{N}, \texorpdfstring{$ \beta $}{b} Circular Ensembles with potential}\label{CircularLoopEqn}
\setcounter{equation}{0}

In this section we establish the set of loop equations from first principles for a general potential satisfying
the assumptions \eqref{assume1} and \eqref{assume2}, and for the parameters $ N\in \N $ and $ {\rm Re}(\kappa)>0 $.
We will assume these conditions henceforth. Our approach is an adaptation of Aomoto's method \cite{Aom_1987}, which
is also detailed in depth in Chapter 4.6 of \cite{For_2010}.
\begin{proposition}\label{LoopEqn}
Under the above assumptions, $ z\in \C^{*} $ and $ z\notin \T $, the first Loop Equation is 
\begin{multline}
   (\kappa-1) \partial_{z}W_{1}(z) - \tfrac{1}{2}\kappa z^{-1}W_{2}(z,z) + \tfrac{1}{2}\kappa z^{-1}\left( N^2-W_{1}(z)^2 \right)
\\
   -P_{1}(z)-V'(z)W_{1}(z) + \left[ \kappa(N-1)+1 \right]\lim_{z\to 0}\frac{W_{1}(z)-W_{1}(0)}{2z} = 0 .
\label{LE:1}
\end{multline}
\end{proposition}
\begin{proof}
The Vandermonde determinant is defined in the standard way
\begin{equation}
   \Delta(\zeta_{1},\ldots,\zeta_{N}) := \prod_{1\leq j<k\leq N}(\zeta_{j}-\zeta_{k}).
\end{equation} 
A key identity under the restriction $ \zeta_{j}= e^{i\theta_j} $, is the analytic re-expression of the squared
modulus of the Vandermonde determinant $ |\zeta_{j}-\zeta_{k}|^2 = (\zeta_{j}-\zeta_{k})(\zeta_{j}^{-1}-\zeta_{k}^{-1}) $.
Let us consider the following definition of $ J_{p} $ and the rewriting of this using integration by parts 
\begin{multline}
 J_{p} := \int_{\mathbb{T}}\frac{d\zeta_{1}}{2\pi i \zeta_{1}} \cdots \int_{\mathbb{T}}\frac{d\zeta_{p}}{2\pi i \zeta_{p}} \cdots \int_{\mathbb{T}}\frac{d\zeta_{N}}{2\pi i \zeta_{N}}
   \frac{\partial}{\partial \zeta_{p}} \left\{ 
                  \frac{\zeta_{p}+z}{\zeta_{p}-z} e^{-\sum_{j}V(\zeta_{j})} |\Delta|^{2\kappa}
                                   \right\}
\\
  = \int_{\mathbb{T}}\frac{d\zeta_{1}}{2\pi i \zeta_{1}} \cdots \mathbb{I}_{p} \cdots \int_{\mathbb{T}}\frac{d\zeta_{N}}{2\pi i \zeta_{N}}
                    \left[ \frac{1}{2\pi i \zeta_{p}}\frac{\zeta_{p}+z}{\zeta_{p}-z} e^{-\sum_{j}V(\zeta_{j})} |\Delta|^{2\kappa}
                    \right]^{\theta_{p}=2\pi}_{\theta_{p}=0}
\\
        + \int_{\mathbb{T}}\frac{d\zeta_{1}}{2\pi i \zeta_{1}} \cdots \int_{\mathbb{T}}\frac{d\zeta_{p}}{2\pi i \zeta_{p}} \cdots \int_{\mathbb{T}}\frac{d\zeta_{N}}{2\pi i \zeta_{N}}
          \frac{1}{\zeta_{p}}\frac{\zeta_{p}+z}{\zeta_{p}-z} e^{-\sum_{j}V(\zeta_{j})} |\Delta|^{2\kappa} .
\label{aomoto}
\end{multline}
Now we consider the various terms arising from the left-hand side of \eqref{aomoto}. Firstly we compute the
derivative of the Vandermonde determinant
\begin{equation}
  \frac{\partial}{\partial \zeta_{p}}\log |\Delta|^{2\kappa} = \kappa \frac{1}{\zeta_{p}} \sum_{1\leq r\neq p\leq N} \frac{\zeta_{p}+\zeta_{r}}{\zeta_{p}-\zeta_{r}} .
\label{Vdetderiv}
\end{equation}
Using this we next sum the left-hand side of \eqref{aomoto} over all independent $ p $ and find
\begin{multline}
   \sum_{p=1}^{N}J_{p} = -2z\Big\langle \sum_{p=1}^{N}(\zeta_{p}-z)^{-2} \Big\rangle
   - \Big\langle \sum_{p=1}^{N}\frac{\zeta_{p}+z}{\zeta_{p}-z}V'(\zeta_{p}) \Big\rangle
\\
   + \kappa\Big\langle \sum_{p=1}^{N}\frac{\zeta_{p}+z}{\zeta_{p}-z}\frac{1}{\zeta_{p}}\sum_{1\leq r\neq p\leq N}\frac{\zeta_{p}+\zeta_{r}}{\zeta_{p}-\zeta_{r}} \Big\rangle .
\label{LHSloop}
\end{multline}
Continuing we seek to express the terms on the right-hand side of \eqref{LHSloop} in terms of the connected
resolvent functions. To this end we note the following averages have such evaluations - starting with
$   \Big\langle \sum_{p=1}^{N}(\zeta_{p}-z)^{-1} \Big\rangle = \frac{\displaystyle 1}{\displaystyle 2z}\left[ W_{1}(z)-N \right] $,
we deduce
$   \Big\langle \sum_{p=1}^{N}\zeta_{p}^{-1}\frac{\displaystyle\zeta_{p}+z}{\displaystyle\zeta_{p}-z} \Big\rangle
    = \frac{\displaystyle 1}{\displaystyle z}\left[ W_{1}(z)-N \right] - \Big\langle \sum_{p=1}^{N}\zeta_{p}^{-1} \Big\rangle $,
and also find
$   2z\Big\langle \sum_{p=1}^{N}(\zeta_{p}-z)^{-2} \Big\rangle
    = \frac{\displaystyle\partial}{\displaystyle\partial z}W_{1}(z)-\frac{\displaystyle 1}{\displaystyle z}\left[ W_{1}(z)-N \right] $.
This latter result gives the first term on the right-hand side of \eqref{LHSloop}. Furthermore, for $ z, z' \notin \T $,
we compute
\begin{equation*}
   4zz'\Big\langle \sum_{p,r=1}^{N}(\zeta_{p}-z)^{-1}(\zeta_{r}-z')^{-1} \Big\rangle
   = W_{2}(z,z')+\left[ W_{1}(z)-N \right]\left[ W_{1}(z')-N \right] .
\end{equation*}
Now we turn our attention to the third term on the right-hand side of \eqref{LHSloop}. From the symmetry
of the integral under $ p \leftrightarrow r $ we deduce
\begin{multline*}
  \Big\langle \sum_{p=1}^{N}\sum_{\substack{r=1 \\ r\neq p}}^{N}\frac{1}{\zeta_{p}}\frac{\zeta_{p}+z}{\zeta_{p}-z}\frac{\zeta_{p}+\zeta_{r}}{\zeta_{p}-\zeta_{r}} \Big\rangle
\\
  =
      \tfrac{1}{2}\Big\langle \sum_{p=1}^{N}\sum_{\substack{r=1 \\ r\neq p}}^{N}\frac{1}{\zeta_{p}}\frac{\zeta_{p}+z}{\zeta_{p}-z}\frac{\zeta_{p}+\zeta_{r}}{\zeta_{p}-\zeta_{r}} \Big\rangle
    + \tfrac{1}{2}\Big\langle \sum_{r=1}^{N}\sum_{\substack{p=1 \\ p\neq r}}^{N}\frac{1}{\zeta_{r}}\frac{\zeta_{r}+z}{\zeta_{r}-z}\frac{\zeta_{r}+\zeta_{p}}{\zeta_{r}-\zeta_{p}} \Big\rangle 
\\
  = - \Big\langle \sum_{\substack{1\leq r,p\leq N \\ p\neq r}}\frac{\zeta_{p}+\zeta_{r}}{(\zeta_{p}-z)(\zeta_{r}-z)} \Big\rangle 
    +  \tfrac{1}{2}\Big\langle \sum_{\substack{1\leq r,p\leq N \\ p\neq r}}\left( \frac{1}{\zeta_{r}}+\frac{1}{\zeta_{p}} \right) \Big\rangle
\\
  = -\frac{1}{2z}\left( W_{2}(z,z)+[W_{1}(z)-N]^2 \right) + \frac{\partial}{\partial z} W_{1}(z)
\\
      + \frac{1}{z}N[N-W_{1}(z)] + (N-1)\Big\langle \sum_{p}\zeta_{p}^{-1} \Big\rangle .
\end{multline*}
The second term on the right-hand side of \eqref{LHSloop} is
$  \Big\langle \sum_{p=1}^{N}\frac{\displaystyle \zeta_{p}+z}{\displaystyle \zeta_{p}-z}V'(\zeta_{p}) \Big\rangle = P_{1}(z)+V'(z)W_{1}(z) $.
Assuming \eqref{assume1} the right-hand side of $ \sum_{p=1}^{N}J_{p} $ in \eqref{aomoto} is given by
$   \frac{\displaystyle 1}{\displaystyle z}[W_{1}(z)-N] - \Big\langle \sum_{p}\zeta_{p}^{-1} \Big\rangle $.
Lastly we can evaluate the average appearing above as 
$  \Big\langle \sum_{p}\zeta_{p}^{-1} \Big\rangle = \lim_{z\to 0}\frac{\displaystyle W_{1}(z)-W_{1}(0)}{\displaystyle 2z} $.
Such a limit exists given the analyticity of $ W_{1}(z) $ for $ z\in \D $. Combining all of these results we arrive
at \eqref{LE:1}.
\end{proof}

Our next objective is to construct the hierarchy of loop equations, of which Proposition \ref{LoopEqn} is
just the base or seed equation. To do this we will employ the insertion operator method \cite{BEMP_2012}, \cite{BMS_2011}
suitably adapted to the unit circle support.
We rewrite potential given in \eqref{potential} using the coefficients $ v_{k}=kt_{k} $ thus
\begin{equation*}
   V(\zeta) = \sum_{\substack{k\in \mathbb{Z} \\ k\neq 0}} k^{-1}v_{k}\zeta^{k}, \quad
   V'(\zeta) = \sum_{\substack{k\in \mathbb{Z} \\ k\neq 0}} v_{k}\zeta^{k-1} .
\end{equation*} 
Employing this new parametrisation we define the insertion operator $ \zeta \in \C^{\star} $
\begin{equation*}
   \frac{\partial}{\partial V(\zeta)} := \sum_{\substack{k\in \mathbb{Z} \\ k\neq 0}} |k|\zeta^{-k}\frac{\partial}{\partial v_{k}} ,
\end{equation*}
which has the following properties -
\begin{enumerate}[(i)]
\item 
if $ \zeta \neq z $ the action on the potential itself is
\begin{equation}
   \frac{\partial}{\partial V(\zeta)}V(z) := \frac{\zeta+z}{\zeta-z} ,
\label{IO_V}
\end{equation}
\item
the derivation of products
\begin{equation}
   \frac{\partial}{\partial V(\zeta)} A[V]\cdot B[V] = \frac{\partial}{\partial V(\zeta)}A[V]\cdot B[V] + A[V]\cdot\frac{\partial}{\partial V(\zeta)}B[V] ,
\label{IO_product}
\end{equation}
\item
satisfies the chain rule for any sufficiently, continuously differentiable function $ f: \mathbb{C} \to \mathbb{C} $
\begin{equation}
   \frac{\partial}{\partial V(\zeta)} f(V(z)) = f^{\prime}(V(z))\frac{\zeta+z}{\zeta-z} ,
\label{IO_f}
\end{equation} 
\item
and commutes with ordinary derivation, $ \zeta \neq z $
\begin{equation}
   \frac{\partial}{\partial V(\zeta)} \frac{\partial}{\partial z} = \frac{\partial}{\partial z} \frac{\partial}{\partial V(\zeta)} .
\label{IO_commute}
\end{equation}
\end{enumerate}

Proceeding on with the task of constructing the higher loop equations we establish a number of preliminary Lemmas.
\begin{lemma}
The first resolvent function is given by
\begin{equation}
  W_{1}(z) = \frac{\partial}{\partial V(z)} \log Z_N ,\quad z\in\C^{\star}\backslash\T ,
\label{W_recur1}
\end{equation}
or recursively with the convention $ W_{0}:=\log Z_N $.
\end{lemma}
\begin{proof}
This, the first case ($ n=1 $) of a sequence, is established by the computation
\begin{align*}
   \frac{\partial}{\partial V(z)} \log Z_N
   & = \frac{1}{Z_{N}} \int \frac{d\zeta_{1}}{2\pi i\zeta_{1}} \cdots \frac{d\zeta_{N}}{2\pi i\zeta_{N}}
                  \prod^{N}_{j=1}e^{-V(\zeta_{j})}\sum^{N}_{l=1}(-1)\frac{\partial}{\partial V(z)}V(\zeta_{l}) |\Delta(\zeta)|^{2\kappa}
\\
   & = \frac{1}{Z_{N}} \int \frac{d\zeta_{1}}{2\pi i\zeta_{1}} \cdots \frac{d\zeta_{N}}{2\pi i\zeta_{N}}
                  \prod^{N}_{j=1}e^{-V(\zeta_{j})}\sum^{N}_{l=1}(-1)\frac{z+\zeta_{l}}{z-\zeta_{l}} |\Delta(\zeta)|^{2\kappa}
\\
   & =  \left \langle \sum^{N}_{l=1}\frac{\zeta_{l}+z}{\zeta_{l}-z} \right \rangle = W_{1}(z) . 
\end{align*}
\end{proof}

\begin{lemma}\label{LemmaUQrecur}
Let $ z\in \C^{\star} $, and $ z_1\in \C^{\star}, \ldots ,z_m\in \C^{\star} $ be pair-wise distinct. The unconnected 
moment $ U_m $ satisfies the recurrence relation for $ m \in \N $
\begin{equation}
   \frac{\partial}{\partial V(z)} U_{m}(z_1,\ldots, z_m) = U_{m+1}(z_1,\ldots,z_m,z)-W_1(z)U_{m}(z_1,\ldots, z_m) .
\label{U_recur}
\end{equation}
Furthermore, with $ z'\in \C^{\star} $ and distinct from the forgoing variables, the unconnected potential moment 
$ Q_{m+1} $ satisfies the recurrence relation
\begin{multline}
   \frac{\partial}{\partial V(z')} Q_{m+1}(z;z_1,\ldots, z_m)
\\
   = Q_{m+2}(z;z_1,\ldots,z_m,z')-W_1(z')Q_{m+1}(z;z_1,\ldots, z_m)
\\
   + \frac{\partial}{\partial z'}\left( \frac{z'+z}{z'-z}U_{m+1}(z_1,\ldots,z_m,z') \right)
   - \frac{1}{z'}\frac{z'+z}{z'-z}U_{m+1}(z_1,\ldots,z_m,z')
\\
   - \frac{N}{z'}U_{m}(z_1,\ldots,z_m) .
\label{Q_recur}
\end{multline}
\end{lemma}
\begin{proof}
Assume that  $ z'\neq z $, $ z_1, \ldots, z_n $ are all pair-wise distinct. Let us define the Riesz-Herglotz kernel sum
$   A(z) := \sum^{N}_{l=1}\frac{\displaystyle \zeta_{l}+z}{\displaystyle \zeta_{l}-z} $,
and the divided-difference potential analogue
$   A_{0}(z) := \sum^{N}_{l=1}\frac{\displaystyle \zeta_{l}+z}{\displaystyle \zeta_{l}-z}\left[ V'(\zeta_{l})-V'(z) \right] $.
For any $ B(\zeta_1,\ldots,\zeta_N) $ composed of products of $ A, A_0 $ we compute that the action of the
insertion operator on its configuration average is the sum of three parts, using \eqref{W_recur1},
\eqref{IO_product} and \eqref{IO_f},
\begin{equation}
   \frac{\partial}{\partial V(z)} \left\langle B \right\rangle
     = -W_{1}(z)\left\langle B \right\rangle + \left\langle B A(z) \right\rangle
       +\left\langle \frac{\partial}{\partial V(z)}B \right\rangle .
\label{etc1}
\end{equation} 
Furthermore, employing \eqref{IO_V} and \eqref{IO_commute}, we compute that
\begin{equation}
   \frac{\partial}{\partial V(z')}A_{0}(z) 
   = -\frac{N}{z'}+\frac{\partial}{\partial z'}\left( \frac{z'+z}{z'-z}A(z') \right)-\frac{1}{z'}\frac{z'+z}{z'-z}A(z') .
\label{etc2}
\end{equation}
Now we proceed to compute the action of the insertion operator on the product $ \left\langle A_{0}(z)A(z_{1})\cdots A(z_{n}) \right\rangle $
by applying the forgoing results. First we apply \eqref{etc1} to this particular product and note that 
$ \frac{\partial}{\partial V(z)}A(z_j)=0 $. Next we substitute \eqref{etc2} into the appropriate term of the
resulting expression and then deduce
\begin{multline}
   \frac{\partial}{\partial V(z')} \left\langle A_{0}(z)A(z_{1})\cdots A(z_{n}) \right\rangle 
   =  \left\langle A_{0}(z)A(z_{1})\cdots A(z_{n})A(z') \right\rangle 
\\
   - \left\langle A_{0}(z)A(z_{1})\cdots A(z_{n}) \right\rangle W_{1}(z')
      -\frac{N}{z'} \left\langle A(z_{1})\cdots A(z_{n}) \right\rangle
\\
      +\frac{\partial}{\partial z'}\left( \frac{z'+z}{z'-z}\left\langle A(z_{1})\cdots A(z_{n})A(z') \right\rangle \right)
      -\frac{1}{z'}\frac{z'+z}{z'-z}\left\langle A(z_{1})\cdots A(z_{n})A(z') \right\rangle .
\label{etc3}
\end{multline}
Both \eqref{U_recur} and \eqref{Q_recur} now follow as applications of the above relation.
\end{proof}

A key result is that the action of the insertion operator on a particular connected resolvent
function generates the next connected resolvent function.
\begin{proposition}\label{IOgenerator}
Let us take the variables $ z_1\in \C^{\star}, \ldots, z_n\in \C^{\star} $ pair-wise distinct.
The resolvent functions $ W_n $, $ n\in\N $ are computed from the generating function using the relation
\begin{equation}
    \frac{\partial}{\partial V(z_{1})} \cdots \frac{\partial}{\partial V(z_{n})} \log Z_N 
      = W_{n}(z_{1}, \ldots, z_{n}) .
\label{IOgenW}
\end{equation}
\end{proposition}
\begin{proof}
To establish this result we will prove it in its recursive form and then appeal to the initial relation
\eqref{W_recur1}. In order to prove the recursive form we consider the action of the insertion operator using 
\eqref{U_recur} in two different ways, firstly in the form
\begin{multline*}
   \frac{\partial}{\partial V(z_{l+2})} U_{l+1}(z_1,\ldots, z_{l+1})
\\
   = U_{l+2}(z_1,\ldots,z_{l+1},z_{l+2})-W_1(z_{l+2})U_{l+1}(z_1,\ldots, z_{l+1}) .
\end{multline*} 
Now we compute the left-hand side of the above starting with the recursive moment-cumulant relation \eqref{U-W_recur}
(here $ I=(z_1, \ldots, z_l) $)
\begin{multline*}
   \frac{\partial}{\partial V(z_{l+2})} U_{l+1}(z_1,\ldots, z_{l+1})
\\
   = \sum_{I_j\subseteq I}\left\{ \frac{\partial}{\partial V(z_{l+2})}W_{l+1-\#(I_j)}(I\backslash I_j\| z_{l+1})U_{\#(I_j)}(I_j) \right. 
\\ \left.
                                 +W_{l+1-\#(I_j)}(I\backslash I_j\| z_{l+1})\frac{\partial}{\partial V(z_{l+2})}U_{\#(I_j)}(I_j) 
                         \right\}
\end{multline*}
\begin{multline*}
   = \frac{\partial}{\partial V(z_{l+2})}W_{l+1}(I\| z_{l+1})
      + \sum_{\substack{I_j\subseteq I \\ I_j\neq \emptyset}}W_{l+2-\#(I_j)}(I\backslash I_j\| z_{l+1},z_{l+2})U_{\#(I_j)}(I_j)
\\
      + \sum_{I_j\subseteq I}W_{l+1-\#(I_j)}(I\backslash I_j\| z_{l+1})\left[ U_{\#(I_j)+1}(I_j\| z_{l+2}) - W_1(z_{l+2}) U_{\#(I_j)}(I_j) \right]
\end{multline*}
\begin{multline*}
   =  \frac{\partial}{\partial V(z_{l+2})}W_{l+1}(I\| z_{l+1})
      + \sum_{\substack{I_j\subseteq I \\ I_j\neq \emptyset}}W_{l+2-\#(I_j)}(I\backslash I_j\| z_{l+1},z_{l+2})U_{\#(I_j)}(I_j)
\\
      + \sum_{I_j\subseteq I}W_{l+1-\#(I_j)}(I\backslash I_j\| z_{l+1})U_{\#(I_j)+1}(I_j\| z_{l+2})
\\
      - W_1(z_{l+2}) \sum_{I_j\subseteq I}W_{l+1-\#(I_j)}(I\backslash I_j\| z_{l+1})U_{\#(I_j)}(I_j)
\end{multline*}
\begin{multline*}
   =  \frac{\partial}{\partial V(z_{l+2})}W_{l+1}(I\| z_{l+1})
      + \sum_{\substack{I_j\subseteq I\| z_{l+2} \\ I_j\neq \emptyset}}W_{l+2-\#(I_j)}(I\| z_{l+2}\backslash I_j\| z_{l+1})U_{\#(I_j)}(I_j)
\\
       - W_1(z_{l+2}) \sum_{I_j\subseteq I}W_{l+1-\#(I_j)}(I\backslash I_j\| z_{l+1})U_{\#(I_j)}(I_j) 
\end{multline*}
\begin{multline*}
   =  \frac{\partial}{\partial V(z_{l+2})}W_{l+1}(I\| z_{l+1})
\\
      + U_{l+2}(z_1,\ldots,z_{l+2})-W_{l+2}(z_1,\ldots,z_{l+2})
      - W_1(z_{l+2})U_{l+1}(z_1,\ldots,z_{l+1}) .
\end{multline*}
In the second step we have used \eqref{IO_product}; in the third \eqref{U_recur}; in the fourth we have noted that
the two terms in the summand are just a division of a common term according to whether $ z_{l+2} $ is either in the
argument of the $ W $ or the $ U $ factor; and the final step is a recognition of the sums involved.
Upon comparing the two expressions we conclude
\begin{equation*}
   \frac{\partial}{\partial V(z_{l+2})} W_{l+1}(z_1,\ldots, z_{l+1}) - W_{l+2}(z_1,\ldots, z_{l+2}) = 0 .
\end{equation*} 
\end{proof}

In addition we require the action of the insertion operator on the potential resolvent functions.
\begin{lemma}\label{IOpotentialCF}
Applying the insertion operator to $ P_{n} $ gives, for $ n=1 $
\begin{equation}
   \frac{\partial}{\partial V(z_1)}P_{1}(z)
   = P_{2}(z;z_1)-\frac{N}{z_1}+\frac{\partial}{\partial z_1}\left( \frac{z_1+z}{z_1-z}W_{1}(z_1) \right)-\frac{1}{z_1}\frac{z_1+z}{z_1-z}W_{1}(z_1) ,
\label{P_recur1}
\end{equation}
and for $ n>1 $
\begin{multline}
   \frac{\partial}{\partial V(z_{n+1})}P_{n+1}(z;z_1,\ldots,z_n) = P_{n+2}(z;z_1,\ldots,z_n,z_{n+1})
\\
   +\frac{\partial}{\partial z_{n+1}}\left( \frac{z_{n+1}+z}{z_{n+1}-z}W_{n+1}(z_1,\ldots,z_{n+1}) \right)
                                      -\frac{1}{z_{n+1}}\frac{z_{n+1}+z}{z_{n+1}-z}W_{n+1}(z_1,\ldots,z_{n+1}) .
\label{P_recur}
\end{multline}
\end{lemma}
\begin{proof}
For \eqref{P_recur1} we apply Lemma \ref{LemmaUQrecur} to the case $ P_1(z)=\left\langle A_0(z) \right\rangle $.
Using \eqref{etc3} and the definition $ \left\langle A_0(z)A(z') \right\rangle = P_2(z;z')+P_1(z)W_1(z') $, 
$ \left\langle A(z') \right\rangle = W_1(z') $ we immediately deduce \eqref{P_recur1}.
In order to prove \eqref{P_recur} we adopt a similar strategy to that employed in the proof of Proposition \ref{IOgenerator}.
Consider the action of the insertion operator on $ Q_{l+1} $ in two different ways, firstly in the form \eqref{Q_recur}
\begin{multline*}
   \frac{\partial}{\partial V(z_{l+1})} Q_{l+1}(z;z_1,\ldots, z_{l})
\\
   = Q_{l+2}(z;z_1,\ldots,z_{l+1})-W_1(z_{l+1})Q_{l+1}(z;z_1,\ldots, z_{l})
\\ 
   + \frac{\partial}{\partial z_{l+1}}\left( \frac{z_{l+1}+z}{z_{l+1}-z}U_{l+1}(z_1,\ldots,z_{l+1}) \right)
   - \frac{1}{z_{l+1}}\frac{z_{l+1}+z}{z_{l+1}-z}U_{l+1}(z_1,\ldots,z_{l+1})
\\
   - \frac{N}{z_{l+1}}U_{l}(z_1,\ldots,z_l) .
\end{multline*} 
Now we compute the left-hand side of the above starting with the recursive moment-cumulant relation \eqref{Q-P_recur}
(again $ I=(z_1, \ldots, z_l) $) in a sequence of steps
\begin{multline*}
   \frac{\partial}{\partial V(z_{l+1})} Q_{l+1}(z;I)
   = \sum_{I_j\subseteq I}\left\{ \frac{\partial}{\partial V(z_{l+1})}P_{l+1-\#(I_j)}(z;I\backslash I_j)U_{\#(I_j)}(I_j) \right. 
\\ \left.
                                 +P_{l+1-\#(I_j)}(z;I\backslash I_j)\frac{\partial}{\partial V(z_{l+1})}U_{\#(I_j)}(I_j) 
                         \right\}
\end{multline*}
\begin{multline*}
   = \frac{\partial}{\partial V(z_{l+1})}P_{l+1}(z;I)
      + \sum_{\substack{I_j\subseteq I \\ I_j\neq \emptyset}}\bigg\{ P_{l+2-\#(I_j)}(z;I\backslash I_j\| z_{l+1})
\\
   + \frac{\partial}{\partial z_{l+1}}\left[ \frac{z_{l+1}+z}{z_{l+1}-z}W_{l+1-\#(I_j)}(I\backslash I_j\| z_{l+1}) \right]
   - \frac{1}{z_{l+1}}\frac{z_{l+1}+z}{z_{l+1}-z}W_{l+1-\#(I_j)}(I\backslash I_j\| z_{l+1})
\\
   - \delta_{I_j=I}\frac{N}{z_{l+1}} \bigg\}U_{\#(I_j)}(I_j)
\\
      + \sum_{I_j\subseteq I}P_{l+1-\#(I_j)}(z;I\backslash I_j)\left[ U_{\#(I_j)+1}(I_j\| z_{l+1}) - W_1(z_{l+1}) U_{\#(I_j)}(I_j) \right]
\end{multline*}
\begin{multline*}
   =  \frac{\partial}{\partial V(z_{l+1})}P_{l+1}(z;I)
      + \sum_{\substack{I_j\subseteq I \\ I_j\neq \emptyset}}P_{l+2-\#(I_j)}(z;I\backslash I_j\| z_{l+1})U_{\#(I_j)}(I_j)
\\
      + \sum_{I_j\subseteq I}P_{l+1-\#(I_j)}(z;I\backslash I_j)U_{\#(I_j)+1}(I_j\| z_{l+1})
\\
   + \frac{\partial}{\partial z_{l+1}}\bigg[ \frac{z_{l+1}+z}{z_{l+1}-z} \sum_{\substack{I_j\subseteq I \\ I_j\neq \emptyset}}W_{l+1-\#(I_j)}(I\backslash I_j\| z_{l+1})U_{\#(I_j)}(I_j) \bigg]
\\
   - \frac{1}{z_{l+1}}\frac{z_{l+1}+z}{z_{l+1}-z}\sum_{\substack{I_j\subseteq I \\ I_j\neq \emptyset}}W_{l+1-\#(I_j)}(I\backslash I_j\| z_{l+1})U_{\#(I_j)}(I_j)
\\
   - \frac{N}{z_{l+1}}U_{l}(I)      
   - W_1(z_{l+1}) \sum_{\substack{I_j\subseteq I \\ I_j\neq \emptyset}}P_{l+1-\#(I_j)}(z;I\backslash I_j)U_{\#(I_j)}(I_j)
\end{multline*}
\begin{multline*}
   =  \frac{\partial}{\partial V(z_{l+1})}P_{l+1}(z;z_1,\ldots,z_{l})
\\
      + Q_{l+2}(z;z_1,\ldots,z_{l+1}) - P_{l+2}(z;z_1,\ldots,z_{l+1}) - P_{l+1}(z;z_1,\ldots,z_{l})U_{1}(z_{l+1})
\\
      + \frac{\partial}{\partial z_{l+1}}\left( \frac{z_{l+1}+z}{z_{l+1}-z}\left[ U_{l+1}(z_1,\ldots,z_{l+1})-W_{l+1}(z_1,\ldots,z_{l+1}) \right] \right)
\\
      - \frac{1}{z_{l+1}} \frac{z_{l+1}+z}{z_{l+1}-z}\left[ U_{l+1}(z_1,\ldots,z_{l+1})-W_{l+1}(z_1,\ldots,z_{l+1}) \right]
\\
      - \frac{N}{z_{l+1}}U_{l}(z_1,\ldots,z_l)
      - W_1(z_{l+1})\left[ Q_{l+1}(z;z_1,\ldots,z_{l})-P_{l+1}(z;z_1,\ldots,z_{l}) \right] .
\end{multline*}
Upon comparing the two expressions we arrive at \eqref{P_recur}.
\end{proof}

Using Proposition \ref{IOgenerator} and Lemma \ref{IOpotentialCF} we can apply the action of the insertion
operator repeatedly to the first Loop Equation.
\begin{proposition}\label{LoopEqn2}
The second Loop Equation $ z\neq z_{1} $, $ z, z_1 \notin \T $ is given by
\begin{multline}
   (\kappa-1) \partial_{z}W_{2}(z,z_{1}) - \tfrac{1}{2}\kappa z^{-1}\left[ W_{3}(z,z,z_{1}) + 2W_{1}(z)W_{2}(z,z_{1}) \right]
\\
   - P_{2}(z;z_{1}) - V'(z)W_{2}(z,z_{1})
\\
   - \frac{\partial}{\partial z_{1}}\left( \frac{z_{1}+z}{z_{1}-z}\left[ W_{1}(z_{1})-W_{1}(z) \right] \right)
   + \frac{1}{z_{1}}\frac{z_{1}+z}{z_{1}-z}\left[ W_{1}(z_{1})-W_{1}(z) \right]
\\
   +\frac{N-W_{1}(z)}{z_{1}} + \left[ \kappa(N-1)+1 \right]\lim_{z\to 0}\frac{W_{2}(z,z_{1})}{2z} = 0 .
\label{LE:2}
\end{multline}
Let $ I $ denote the $m$-tuple of variables $ I=(z_{1},z_{2}, \ldots,z_{m}) $ and $ \| $ the string concatenation operation.
In the general case the $(m+1)$-th Loop Equation for $ m\geq 2 $ is
\begin{multline}
   (\kappa-1)\partial_{z}W_{m+1}(z\| I)
\\
   - \tfrac{1}{2}\kappa z^{-1} \Big( W_{m+2}(z,z\| I)
             +\sum_{\displaystyle\substack{I_{j}\subset I \\ 0\leq j=|I_{j}|\leq m}} W_{j+1}(z\| I_{j})W_{m-j+1}(z\| I\backslash I_{j})
                               \Big)
\\
   - P_{m+1}(z;I) - V'(z)W_{m+1}(z\| I) 
\\
   - \sum^{m}_{j=1} \frac{\partial}{\partial z_{j}}\left( \frac{z_{j}+z}{z_{j}-z}\left[ W_{m}(I)-W_{m}(z\| I\backslash z_{j}) \right] \right)
\\
   + \sum^{m}_{j=1} \frac{1}{z_{j}}\frac{z_{j}+z}{z_{j}-z}\left[ W_{m}(I)-W_{m}(z\| I\backslash z_{j}) \right]
\\
   - \sum^{m}_{j=1} \frac{1}{z_{j}}W_{m}(z\| I\backslash z_{j}) + \tfrac{1}{2}[\kappa N+1-\kappa]\lim_{z\to 0}\frac{W_{m+1}(z\| I)}{z} = 0 .
\label{LE:=>3}
\end{multline}
\end{proposition}
\begin{proof}
In respect of the second loop equation \eqref{LE:2} we apply the insertion operator $ \partial/\partial V(z_1) $
to \eqref{LE:1} assuming $ z\neq z_1 $. Employing \eqref{IOgenW}, \eqref{IO_commute}, \eqref{IO_product}, \eqref{P_recur1}
and \eqref{IO_V}, and interchanging the $ z\to 0 $ limit in the resulting expression and simplifying we deduce \eqref{LE:2}. 
To prove the generic case \eqref{LE:=>3}, which applies for $ m+1\geq 3 $, we are going to employ an induction argument
and utilise all of our previous lemmas. We act on the left-hand side of $(m+1)$-th loop equation \eqref{LE:=>3} with
the insertion operator $ \partial/\partial V(z_{m+1}) $ and note the following mappings of the terms (now $ \hat{I}=(z_1,\ldots,z_{m+1}) $)
\begin{equation*}
  W_{m}(I) \mapsto W_{m+1}(\hat{I}) ,
\end{equation*}
\begin{equation*}
  W_{m+1}(z\| I) \mapsto W_{m+2}(z\|\hat{I}) ,
\end{equation*}
\begin{equation*}
  W_{m+2}(z,z\| I) \mapsto W_{m+3}(z,z\|\hat{I}) ,
\end{equation*}
\begin{equation*}
   W_{j+1}(z\| I_j) \mapsto W_{j+2}(z\|I_j\| z_{m+1}) ,
\end{equation*}
\begin{equation*}
   W_{m-j+1}(z\| I\backslash I_j) \mapsto W_{m-j+2}(z\|I\backslash I_j\| z_{m+1}) ,
\end{equation*}
\begin{multline*}
   P_{m+1}(z;I) \mapsto P_{m+2}(z;\hat{I})  
\nonumber\\
   +\frac{\partial}{\partial z_{m+1}}\left( \frac{z_{m+1}+z}{z_{m+1}-z}W_{m+1}(\hat{I}) \right)
   -\frac{1}{z_{m+1}}\frac{z_{m+1}+z}{z_{m+1}-z}W_{m+1}(\hat{I}) , 
\end{multline*}
\begin{equation*}
   W_{m}(z\| I\backslash z_j) \mapsto W_{m+1}(z\|\hat{I}\backslash z_{j}) ,\quad j\neq m+1 ,
\end{equation*}
\begin{equation*}
  \partial_z W_{m+1}(z\| I) \mapsto \partial_z W_{m+2}(z\|\hat{I}) ,
\end{equation*}
\begin{equation*}
  -\frac{\partial}{\partial V(z_{m+1})}V'(z) = \frac{\partial}{\partial z_{m+1}}\frac{z_{m+1}+z}{z_{m+1}-z}-\frac{1}{z_{m+1}}\frac{z_{m+1}+z}{z_{m+1}-z}-\frac{1}{z_{m+1}} .
\end{equation*}
From the fourth and fifth mappings in this list we note that
\begin{multline*}
  \sum_{I_j\subseteq I} W_{j+1}(z\| I_j)W_{m-j+1}(z\| I\backslash I_j) \mapsto
\\
  \sum_{I_j\subseteq I} W_{j+2}(z\|I_j\| z_{m+1})W_{m-j+1}(z\| I\backslash I_j)+W_{j+1}(z\| I_j)W_{m-j+2}(z\|I\backslash I_j\| z_{m+1})
\\
  = \sum_{I_j\subseteq \hat{I}} W_{j+1}(z\| I_j)W_{m-j+2}(z\| \hat{I}\backslash I_j) ,
\end{multline*}
where we recognise the two terms in the intermediate summation as arising from the latter as to whether 
$ z_{m+1}\in I_j $ or not. Combining all these and sorting terms into appropriate categories we see that the
the resulting expression is precisely the $(m+2)$-th loop equation.
\end{proof}

\section{Large \texorpdfstring{$ N $}{N} solution scheme for Loop Equations for general \texorpdfstring{$ N $}{N} and \texorpdfstring{$ \beta $}{b} for the Dyson Circular Ensemble in the Global Regime}\label{CircularSolnScheme}
\setcounter{equation}{0}

Two asymptotic regimes of the general system of loop equations as $ N\to \infty $ are permissible. One regime,
which we refer to as a {\it Continuum Limit}, is the regime where the index $ k $ of the
moments $ m_k $ grows like $ k\to \infty $ but with fixed $ k/N = x $ so that $ x={\rm O}(1) $.
The moments have the limit
\begin{equation}
     m(x) := \lim_{N\to \infty}\frac{1}{N}m_{k=xN} .  
\label{CLimit}
\end{equation} 
This regime requires a careful analysis of the jumps in $ W_2(\zeta) $ across the unit circle $ \zeta\in\T $
and of the densities on the unit circle which contain terms that are purely oscillatory with phases proportional
to $ N $, such as $ \zeta^N $ (in addition to the purely algebraic dependency on $ N $). This essentially implies a 
local analysis in the neighbourhood of distinguished or singular points on the unit circle and a new independent
variable replacing $ \zeta $, depending on the details of the potential.

The other regime is when either $ |\zeta|<1 $ or $ |\zeta|>1 $, i.e. bounded away from the unit circle, and thus $ \zeta^N $
is exponentially suppressed or dominant depending on the situation - we denote this the {\it Global Regime}. In this
case the moment index $ k={\rm O}(1) $ is fixed or $ k={\rm o}(N) $, and no information about the larger values of 
$ k \sim {\rm O}(N) $ is apriori accessible. This is the only case we will study here. Nonetheless, by taking
$ N,k\to \infty $ such that $ k/N $ is fixed in the resulting expressions, we can reclaim the expansion \eqref{Sa}.
This is consistent with $ f(k;\beta) $ as defined in \eqref{S2} being analytic in $ k $ with radius of convergence
$ \min(2\pi, 2\beta) $.

For the Circular $\beta$ Ensemble in the global regime it is possible to use elementary arguments to fix the algebraic
growth of the cumulants, which we do in the following proposition.

\begin{proposition}\label{cumulantgrowth}
In the global regime, $ ||z_j|-1| > \delta $, $ j=1, \ldots, l $, $ 1 > \delta > 0 $ and all $ l\geq 1 $,
$ {\rm Re}(\kappa)>0 $ as $ N \to \infty $ the connected resolvent functions $ W_l, P_l $, $ l \geq 1 $ have
algebraic leading order and possess the large $ N $ expansion
\begin{align}
    W_{l} & = N^{2-l}W^{(2-l)}_{l} + N^{1-l}W^{(1-l)}_{l}+\ldots ,
\label{globalX:a} \\
    P_{l} & = N^{2-l}P^{(2-l)}_{l} + N^{1-l}P^{(1-l)}_{l}+\ldots .
\label{globalX:b}
\end{align}
\end{proposition}
\begin{proof}
We will show this for the $ W_l $ only as the arguments are identical in the case of the $ P_l $.
For any $ z\in \C^{\star} $ such that $ ||z|-1| > \delta $ and $ \zeta \in \T $ we note the following bounds
using the triangle inequality
\begin{equation*}
   \frac{|1-|z||}{1+|z|} \leq \left| \frac{\zeta+z}{\zeta-z} \right| \leq \frac{1+|z|}{|1-|z||} .
\end{equation*}
These bounds apply for all $ z\in \C^{\star} $ excluded from the annulus $ \{z\in \C: 1-\delta < |z| < 1+\delta \}$
and thus we do not need to keep track of the configurations of the co-ordinates $ (z_1, \ldots, z_l) $. 
Applying these basic inequalities to the integral definition of $ U_l $, we have
\begin{equation*}
   \prod_{1\leq i\leq l}\frac{\left||z_i|-1\right|}{|z_i|+1}N^l
   \leq |U_l| \leq \prod_{1\leq i\leq l}\frac{|z_i|+1}{\left||z_i|-1\right|}N^l .
\end{equation*} 
Therefore the $ U_l $ have algebraic growth and because of the purely polynomial relationship with the
$ W_l $ (the inverse of \eqref{moment-cumulant}) the same conclusion can be drawn for them. However in order
to refine the large $ N $ behaviour of the $ W_l $ we will make an analysis of \eqref{LE:=>3} using balancing 
arguments. Let us denote the leading order algebraic term by $ W_l = {\rm O}(N^{E_l}) $ with the exponent
$ E_l $. There are five types of terms in \eqref{LE:=>3} with distinct exponents -
\begin{enumerate}
 \item 
 $ \mathfrak{A} $: terms $ W_{l+2} $, with exponent $ E_{l+2} $,
 \item
 $ \mathfrak{B} $: terms $ \partial_z W_{l+1}, W_{l+1}, P_{l+1} $, with exponent $ E_{l+1} $,
 \item
 $ \mathfrak{C} $: terms $ W_{l} $, with exponent $ E_{l} $,
 \item
 $ \mathfrak{D} $: terms $ N W_{l+1} $, with exponent $ E_{l+1}+1 $,
 \item
 $ \mathfrak{F}_{j} $, $ 0\leq j\leq l$: terms $ W_{l+1-j}W_{j+1} $, with exponent $ E_{l+1-j}+E_{j+1} $.
\end{enumerate}
Of the total number of matchings to apply the balancing conditions, the fifth Bell number $ B_5=52 $, a number are
obviously logically inconsistent, such as $ \mathfrak{B} $ and $ \mathfrak{D} $, of which there are sixteen of these.
In addition a further eight are also inconsistent. The single case of no conditions can also be excluded.
A further seven cases lead to $ E_l=0 $ which is just the
original loop equation. A similar set are the eight neutral or fixed cases where $ E_l $ is $ l $ independent however
these are not relevant here. The remaining twelve have potential applications. Of these four are ascending $ E_{l+1}>E_{l} $,
four are descending $ E_{l+1}<E_{l} $ and another four are progressive $ E_{l+1}\lessgtr E_{l} $, depending on the sign of 
$ E_1, E_2 $, or $ E_2-1 $. In all these twelve cases the $ l $ dependence is linear. The descending cases are only
of interest here and are -
\begin{itemize}
 \item[\textbullet] $ \{ \mathfrak{C},\mathfrak{D} \}> \{ \mathfrak{B} \}> \{ \mathfrak{A},\mathfrak{F}_j \} $, $ E_l=-l $,
 \item[\textbullet] $ \{ \mathfrak{C},\mathfrak{D} \}> \{ \mathfrak{B},\mathfrak{F}_j \}> \{ \mathfrak{A} \} $, $ E_l=1-l $,
 \item[\textbullet] $ \{ \mathfrak{C},\mathfrak{D},\mathfrak{F}_j \}> \{ \mathfrak{B} \}> \{ \mathfrak{A} \} $, $ E_l=2-l $,
 \item[\textbullet] $ \{ \mathfrak{C},\mathfrak{D} \}> \{ \mathfrak{B} \}> \{ \mathfrak{A} \}, \{ \mathfrak{F}_j \} $, $ E_l=2E_1-l $ .
\end{itemize}
The last two cases are the same for $ E_1=1 $ and is the solution we are seeking as the others do not ensure
the initial instance $ W_1={\rm O}(N) $. Taking $ E_l=2-l $ we now seek the sub-leading term 
$  W_l = N^{E_l}W^0_l+N^{E_l+\delta_l}W^1_l+{\rm o}(N^{E_l+\delta_l}) $ where $ \delta_l<0 $. Matching the
sub-leading terms from $ \mathfrak{C},\mathfrak{D},\mathfrak{F}_j $ the only solution is $ \delta_l=-1 $,
which also means that the remainder terms left over from the leading one come in at this level.
\end{proof}

We now specialise all of the preceding theory to the Dyson circular ensemble case with $ V(z)=0 $.
In this work our focus will be on the two-point correlation function for the Dyson circular $ \beta $ ensemble
analytically continued in the complex plane in the parameters $ \beta=2\kappa $ and $ N $. From its
definition \eqref{R1b} one can readily deduce that for $ N\geq 2$ a $ (N\!-\!2) $-dimensional integral representation for this correlation function
with the well-known form
\begin{multline}
   \rho_{(2)}(\theta_2,\theta_1) = \frac{N(N-1)}{(2\pi)^{N}}\frac{\Gamma(\kappa+1)^{N}}{\Gamma(\kappa N)}|e^{i\theta_2}-e^{i\theta_1}|^{2\kappa}
\\        \times
          \int_{[0,2\pi]^{N-2}} d\phi_1 \ldots d\phi_{N-2} \prod^{N-2}_{j=1}\prod^{2}_{k=1}|1-e^{i(\phi_j-\theta_k)}|^{2\kappa} \prod_{1\leq j<k\leq N-2}|e^{i\phi_j}-e^{i\phi_k}|^{2\kappa} , 
\label{DCErho2}
\end{multline}
(see Eq.~(13.32) of \cite{For_2010}), where use has been made of the closed form evaluation of the normalisation
as conjectured in Dyson's original paper \cite{Dys_1962d},
\begin{equation}
   Z_N = \frac{\Gamma(1+N\kappa)}{(\Gamma(1+\kappa))^N} ,
\label{DCEnorm}
\end{equation} 
(see e.g. Prop.~4.7.2 of \cite{For_2010})

Because $ V=0 $ and thus $ P_{n}=0, n\geq 1 $ there is rotational symmetry of the ensemble and the one-particle
density is uniform
\begin{equation*}
   \rho_{(1)}(z) = \rho_{(1)}(1) = N , \quad
   \rho_{l} = 
   \begin{cases}\displaystyle
      N , &  l=0 \\ 
                \displaystyle
      0 , &  l\neq 0 \\                                      
   \end{cases} .
\end{equation*} 
Therefore we have 
\begin{equation*}
  W_{1}(z) =
   \begin{cases}\displaystyle
      N , &  z\in \D \\ 
                \displaystyle
     -N , &  z\in \compD \\                                      
   \end{cases} .
\end{equation*} 
All dependency of the higher $ n\geq2 $ resolvent functions on angles is via their differences and for $ n=2 $
we denote $ \theta=\theta_2-\theta_1 $. Let us define the Fourier coefficients of $ \rho_{(2)C}(\theta) $ through
the trigonometric expansion
\begin{equation}
   \rho_{(2)C}(z) = \sum_{k\in \mathbb{Z}} m_{k}z^{k} .
\label{mDefn}
\end{equation} 
They possess an evenness property $ m_{-k} = m_{k} $.
We can see, either from their definition or from the Loop Equations, that $ W_{2}(z,z) = 0 $ for $ z\in\D $
and $ z\in\compD $. The first Loop Equation \eqref{LE:1} is satisfied by 
\begin{equation}
   W_{2}(z_{1},z_{2}) = W_{2}(\zeta=z_{2}/z_{1}) =
   \begin{cases}
      m_{0}+N, & (0,0) \\
      m_{0}+N, & (\infty,\infty) \\
     -m_{0}-N - 4\sum_{k=1}^{\infty}(m_{k}+N)\zeta^{k}, & (\infty,0) \\
     -m_{0}-N - 4\sum_{k=1}^{\infty}(m_{-k}+N)\zeta^{-k}, & (0,\infty)
   \end{cases} .
\label{DysonW2}
\end{equation} 
In addition to the generic symmetry properties \eqref{SYM-labels}, \eqref{SYM-permute}, \eqref{SYM-reduce}, \eqref{SYM-special}
we have special ones for the Dyson circular ensembles -
\begin{enumerate}
\item[(v)]
Let $ \iota $ be the inversion or flipping operator $ \iota : d\mapsto 1/d, z \mapsto z^{-1} $.
Then inversion symmetry is valid in the global regime
\begin{equation}
   W_{n}\left(\begin{array}{ccc} d_1, & \ldots, & d_n \\ z_1^{-1}, & \ldots, & z_n^{-1} \end{array}\right)
 = W_{n}\left(\begin{array}{ccc} \iota(d_1), & \ldots, & \iota(d_n) \\ z_1, & \ldots, & z_{n} \end{array}\right)
\label{SYM-invert}
\end{equation}
\item[(vi)]
and the affine property $ \alpha \neq 0, \infty $
\begin{equation}
   W_{n}\left(\begin{array}{ccc} d_1, & \ldots, & d_n \\ \alpha z_1, & \ldots, & \alpha z_n \end{array}\right)
 = W_{n}\left(\begin{array}{ccc} d_1, & \ldots, & d_n \\ z_1, & \ldots, & z_{n} \end{array}\right) .
\label{SYM-affine}
\end{equation}
\end{enumerate}

We now undertake the task of solving the hierarchy of loop equations, \eqref{LE:1}, \eqref{LE:2} and \eqref{LE:=>3},
using the large $ N $ expansion of the resolvent functions given by \eqref{globalX:a}, starting with
the leading order contributions.

\noindent\underline{$ W^{(1)}_1 $:}
The first Loop Equation decomposes into the separate equations, the first of these arising at order $ N^2 $, and is, 
assuming $ \kappa\neq 0 $
\begin{equation*}
   z^{-1}\left[ 1-(W_{1}^{(1)}(z))^2 \right]+\lim_{z\to 0}\frac{W_{1}^{(1)}(z)-1}{z} = 0 ,
\end{equation*} 
which has the solutions $ W_{1}^{(1)}(z) = \pm 1 $. Clearly $ W_{1}^{(1)}(z) = 1 $, $ z\in\D $ and $ W_{1}^{(1)}(z)=-1 $,
$ z\in \compD $.

\noindent\underline{$ W^{(0)}_1 $:}
The next equation arises at order $ N $ and is, under the same assumptions and from the solutions above
\begin{equation*}
   -z^{-1}W_{1}^{(1)}(z)W_{1}^{(0)}(z)+\tfrac{1}{2}\lim_{z\to 0}\frac{W_{1}^{(0)}(z)}{z} = 0 .   
\end{equation*} 
We deduce that $ W_{1}^{(0)}(z)=0, z\in\D $ and consequently also that $ W_{1}^{(0)}(z)=0, z\in\compD $.

\noindent\underline{$ W^{(-2k-1)}_1 $:}
In general for the case of even orders $ N^{-2k} $, $ k\geq 0 $ we find
\begin{multline}
 (\kappa-1)\partial_{z}W_{1}^{(-2k)}(z)-\tfrac{1}{2}\kappa z^{-1}W_{2}^{(-2k)}(z,z)
\\
  -\kappa z^{-1}\left[ \tfrac{1}{2}(W_{1}^{(-k)}(z))^2+W_{1}^{(-k+1)}(z)W_{1}^{(-k-1)}(z)+\ldots+W_{1}^{(1)}(z)W_{1}^{(-2k-1)}(z) \right]
\\
  +\tfrac{1}{2}(1-\kappa)\lim_{z\to 0}\frac{W_{1}^{(-2k)}(z)}{z}+\tfrac{1}{2}\kappa\lim_{z\to 0}\frac{W_{1}^{(-2k-1)}(z)}{z} = 0 ,
\label{W1evenX}
\end{multline}
which clearly has a unique solution for $ W_{1}^{(-2k-1)}(z) $, given other inputs and that $ \kappa\neq 0 $ and $ W_{1}^{(1)}(z)\neq 0,1/2 $.

\noindent\underline{$ W^{(-2k-2)}_1 $:}
Whereas for the odd orders $ N^{-2k-1} $ we have
\begin{multline}
 (\kappa-1)\partial_{z}W_{1}^{(-2k-1)}(z)-\tfrac{1}{2}\kappa z^{-1}W_{2}^{(-2k-1)}(z,z)
\\
  -\kappa z^{-1}\left[ W_{1}^{(-k)}(z)W_{1}^{(-k-1)}(z)+\ldots+W_{1}^{(1)}(z)W_{1}^{(-2k-2)}(z) \right]
\\
  +\tfrac{1}{2}(1-\kappa)\lim_{z\to 0}\frac{W_{1}^{(-2k-1)}(z)}{z}+\tfrac{1}{2}\kappa\lim_{z\to 0}\frac{W_{1}^{(-2k-2)}(z)}{z} = 0 , 
\label{W1oddX}
\end{multline}
which also has a unique solution for $ W_{1}^{(-2k-2)}(z) $ given the above conditions.

\noindent\underline{$ W^{(0)}_2 $:}
The second Loop Equation generates an equation at the leading order of $ N $ which states
\begin{multline*}
  -\kappa z^{-1}W_{1}^{(1)}(z)W_{2}^{(0)}(z,z_1)
\\
   -\frac{\partial}{\partial z_1}\frac{z_1+z}{z_1-z}\left[ W_{1}^{(1)}(z_1)-W_{1}^{(1)}(z) \right]
   +\frac{1}{z_1}\frac{z_1+z}{z_1-z}\left[ W_{1}^{(1)}(z_1)-W_{1}^{(1)}(z) \right]
\\
  + \frac{1-W_{1}^{(1)}(z)}{z_1}+\tfrac{1}{2}\kappa\lim_{z\to 0}\frac{W_{2}^{(0)}(z,z_1)}{z} = 0 . 
\end{multline*}
In order to solve this we have to consider the domains of $ z, z_1 $ in a particular order - firstly we
choose $ z\in\D, z_1\in\D $ (denoted by $ 0,0 $) which allows us to fix $ \partial_{z}W_{2}^{(0)}(0,z_1)=0 $, 
and thus $ W_{2}^{(0)}(z,z_1)=0 $. Next we consider $ z\in \compD, z_1\in\D $ i.e. ($ \infty,0 $) and from the 
previous derivative evaluation we conclude
\begin{equation}
   W_{2}^{(0)}(z,z_1) = -\frac{4}{\kappa}\frac{zz_1}{(z_1-z)^2} .
\label{koebe}
\end{equation} 
Proceeding we look at the domain $ 0,\infty $, where $ z\in\D, z_1\in \compD $, and initially compute the
derivative at the origin to be $ \partial_{z}W_{2}^{(0)}(0,z_1)=-4/\kappa z_1 $. This allows us to solve 
for $ W_{2}^{(0)}(z,z_1) $ and we obtain the same result as above. The reason why this is the same is because of
the symmetry $ W_2(z\in 0,z_1\in \infty) = W_2(z^{-1}\in \infty,z_1^{-1}\in 0) $. Lastly we examine the
$ \infty,\infty $ domain, and using the previous derivative evaluation we deduce that $ W_{2}^{(0)}(z,z_1)=0 $. 

\noindent\underline{$ W^{(-1)}_2 $:}
At the next order, $ N^0 $, one can derive the equation
\begin{multline*}
  (\kappa-1)\partial_{z}W_{2}^{(0)}(z,z_1)
  -\kappa z^{-1}\left[ W_{1}^{(0)}(z)W_{2}^{(0)}(z,z_1)+W_{1}^{(1)}(z)W_{2}^{(-1)}(z,z_1) \right]
\\
   -\frac{\partial}{\partial z_1}\frac{z_1+z}{z_1-z}\left[ W_{1}^{(0)}(z_1)-W_{1}^{(0)}(z) \right]
   +\frac{1}{z_1}\frac{z_1+z}{z_1-z}\left[ W_{1}^{(0)}(z_1)-W_{1}^{(0)}(z) \right]
\\
  -\frac{W_{1}^{(0)}(z)}{z_1}+\tfrac{1}{2}(1-\kappa)\lim_{z\to 0}\frac{W_{2}^{(0)}(z,z_1)}{z}+\tfrac{1}{2}\kappa\lim_{z\to 0}\frac{W_{2}^{(-1)}(z,z_1)}{z} = 0 . 
\end{multline*}
Again this has a unique solution for $ W_{2}^{(-1)}(z,z_1) $. 

\noindent\underline{$ W^{(-k-1)}_2 $:}
Next we come to the generic case at order $ N^{-k} $ where $ k\in\N $
\begin{multline*}
  (\kappa-1)\partial_{z}W_{2}^{(-k)}(z,z_1)-\tfrac{1}{2}\kappa z^{-1}W_{3}^{(-k)}(z,z,z_1)
\\
  -\kappa z^{-1}\left[ W_{1}^{(-k)}(z)W_{2}^{(0)}(z,z_1)+W_{1}^{(-k+1)}(z)W_{2}^{(-1)}(z,z_1)+ \right. 
\\ \left. \ldots +W_{1}^{(1)}(z)W_{2}^{(-k-1)}(z,z_1) \right]
\\
   -\frac{\partial}{\partial z_1}\frac{z_1+z}{z_1-z}\left[ W_{1}^{(-k)}(z_1)-W_{1}^{(-k)}(z) \right]
   +\frac{1}{z_1}\frac{z_1+z}{z_1-z}\left[ W_{1}^{(-k)}(z_1)-W_{1}^{(-k)}(z) \right]
\\
  -\frac{W_{1}^{(-k)}(z)}{z_1}+\tfrac{1}{2}(1-\kappa)\lim_{z\to 0}\frac{W_{2}^{(-k)}(z,z_1)}{z}+\tfrac{1}{2}\kappa\lim_{z\to 0}\frac{W_{2}^{(-k-1)}(z,z_1)}{z} = 0 . 
\end{multline*}
In this case one solves for $ W_{2}^{(-k-1)}(z,z_1) $.

For the third and higher Loop Equations a generic pattern has set in, so we only treat this general case.
In addition to the simple and general statements about the initial coefficients we can give three known 
exact cases.

\begin{proposition}
The coefficients $ W^{(l)}_{1}(z) $ for $ l \leq 0 $ and $ z\in\D $ or $ z\in\compD $ all vanish.
\end{proposition}
\begin{proof}
This follows by induction from \eqref{W1evenX} and \eqref{W1oddX} given that $ W^{(1)}_1(z)=\pm 1 $ and
$ W^{(l)}_2(z,z)=0 $ for all $ z $ and $ l $.
\end{proof}

\begin{proposition}
The leading coefficients $ W^{(1-l)}_{l+1} $ for $ l \geq 2 $ and all arguments $ z_1, \ldots, z_{l+1} $ vanish.
Thus the leading order of the expansion for $ W_{m}, m=3,4,\ldots $ is one less than that assumed in \eqref{globalX:a}. 
\end{proposition}
\begin{proof}
Let $ I=(z_1,\ldots z_{m}) $. The $(m+1)$-th loop equation resolved to the level $ N^{2-m} $ states
\begin{multline}
   -\tfrac{1}{2}\kappa z^{-1}\sum_{j=0}^{m}\sum_{I_{j}\| I_{m-j}=I} W^{(1+j-m)}_{m+1-j}(z\| I_{m-j})W^{(1-j)}_{j+1}(z\| I_{j})
\\
   -\sum^m_{j=1} \frac{\partial}{\partial z_j}\frac{z_j+z}{z_j-z}\left[ W^{(2-m)}_{m}(I)-W^{(2-m)}_{m}(z\| I\backslash z_j) \right]
\\
   +\sum^m_{j=1} \frac{1}{z_j}\frac{z_j+z}{z_j-z}\left[ W^{(2-m)}_{m}(I)-W^{(2-m)}_{m}(z\| I\backslash z_j) \right]
\\
   -\sum^m_{j=1} \frac{1}{z_j}W^{(2-m)}_{m}(z\| I\backslash z_j)+\tfrac{1}{2}\kappa\lim_{z\to 0}\frac{W^{(1-m)}_{m+1}(z\| I)}{z} = 0 .
\label{etc4}
\end{multline}
It is a non-trivial fact that the Koebe solution $ W^{(0)}_2 $ satisfies the following functional-differential
equation for all configurations of $ z,z_1,z_2 $
\begin{multline*}
   -\kappa z^{-1} W^{(0)}_{2}(z,z_1)W^{(0)}_{2}(z,z_2)
   -\frac{\partial}{\partial z_1}\frac{z_1+z}{z_1-z}\left[ W^{(0)}_{2}(z_1,z_2)-W^{(0)}_{2}(z,z_2) \right]
\\
   -\frac{\partial}{\partial z_2}\frac{z_2+z}{z_2-z}\left[ W^{(0)}_{2}(z_1,z_2)-W^{(0)}_{2}(z,z_1) \right]
\\
   +\frac{1}{z_1}\frac{z_1+z}{z_1-z}\left[ W^{(0)}_{2}(z_1,z_2)-W^{(0)}_{2}(z,z_2) \right]
   +\frac{1}{z_2}\frac{z_2+z}{z_2-z}\left[ W^{(0)}_{2}(z_1,z_2)-W^{(0)}_{2}(z,z_1) \right]
\\
   -\frac{1}{z_1}W^{(0)}_{2}(z,z_2)-\frac{1}{z_2}W^{(0)}_{2}(z,z_1) = 0 , 
\end{multline*}
as one can verify. However the above is just the $ m=2 $ case of \eqref{etc4} with the exception of the 
terms $ -\kappa z^{-1}W^{(-1)}_3(z,z_1,z_2)+\frac{1}{2}\kappa \partial_z W^{(-1)}_3(z,z_1,z_2)|_{z=0} $,
whose unique solution is $ W^{(-1)}_3(z,z_1,z_2)=0 $. For the $ m=3 $ case of \eqref{etc4} one derives an
identical equation for $ W^{(-2)}_4 $, possessing again a null solution, and so on.
\end{proof}

\begin{proposition}
Let $ I=(z_1,\ldots,z_l) $ for $ l\geq 2 $.
The sub-leading coefficients $ W^{(-l)}_{l+1}(z\| I) $ for $ l \geq 2 $ and configuration $ (\infty\, 0^l) $ satisfy
the linear functional relation
\begin{equation}
   W^{(-l)}_{l+1}(z,z_1,\ldots,z_l) = -\frac{2}{\kappa}\sum^{l}_{j=1} \frac{z z_j}{(z-z_j)^2}  W^{(1-l)}_{l}(z,I\backslash z_j) ,
\label{1-row}
\end{equation}
subject to the initial values
\begin{gather*}
   W^{(0)}_1(z) = 0, \quad
   W^{(-1)}_{2}(z,z_1) = -4\frac{\kappa-1}{\kappa^2}\frac{z z_1(z+z_1)}{(z-z_1)^3} .
\end{gather*}
\end{proposition}
\begin{proof}
The $(m+1)$-th loop equation resolved to the level $ N^{1-m} $ states
\begin{multline*}
   (\kappa-1)\partial_z W^{(1-m)}_{m+1}(z\| I)
\\
   -\kappa z^{-1} \sum_{j=0}^{\text{$m/2$ or $(m-1)/2$}}\sum_{I_{j}\| I_{m-j}=I}\left[ W^{(1+j-m)}_{m+1-j}(z\| I_{m-j})W^{(-j)}_{j+1}(z\| I_{j}) \right. 
\\ \left.
                                                                                      +W^{(j-m)}_{m+1-j}(z\| I_{m-j})W^{(1-j)}_{j+1}(z\| I_{j}) \right] 
\\
   -\sum^m_{j=1} \frac{\partial}{\partial z_j}\frac{z_j+z}{z_j-z}\left[ W^{(1-m)}_{m}(I)-W^{(1-m)}_{m}(z\| I\backslash z_j) \right]
\\
   +\sum^m_{j=1} \frac{1}{z_j}\frac{z_j+z}{z_j-z}\left[ W^{(1-m)}_{m}(I)-W^{(1-m)}_{m}(z\| I\backslash z_j) \right]
\\
   -\sum^m_{j=1} \frac{1}{z_j}W^{(1-m)}_{m}(z\| I\backslash z_j)+\tfrac{1}{2}\kappa\lim_{z\to 0}\frac{W^{(-m)}_{m+1}(z\| I)}{z}+\tfrac{1}{2}(1-\kappa)\lim_{z\to 0}\frac{W^{(1-m)}_{m+1}(z\| I)}{z} = 0 .
\end{multline*}
From the previous theorem we know $ W^{(1-m)}_{m+1}=0 $.
With $ (z,z_1, \ldots, z_m)\in (\infty\,0^m) $ there are a number of additional simplifications: 
$ W^{(1-m)}_{m}(I)=0 $ and $ W^{(1)}_{1}(z)=-1 $. Solving for $ W^{(-m)}_{m+1} $ in terms of $ W^{(1-m)}_{m} $,
using the Koebe solution \eqref{koebe} for $ W^{(0)}_{2} $ and simplifying, we deduce \eqref{1-row}.  
\end{proof}

\begin{figure*}[h!]
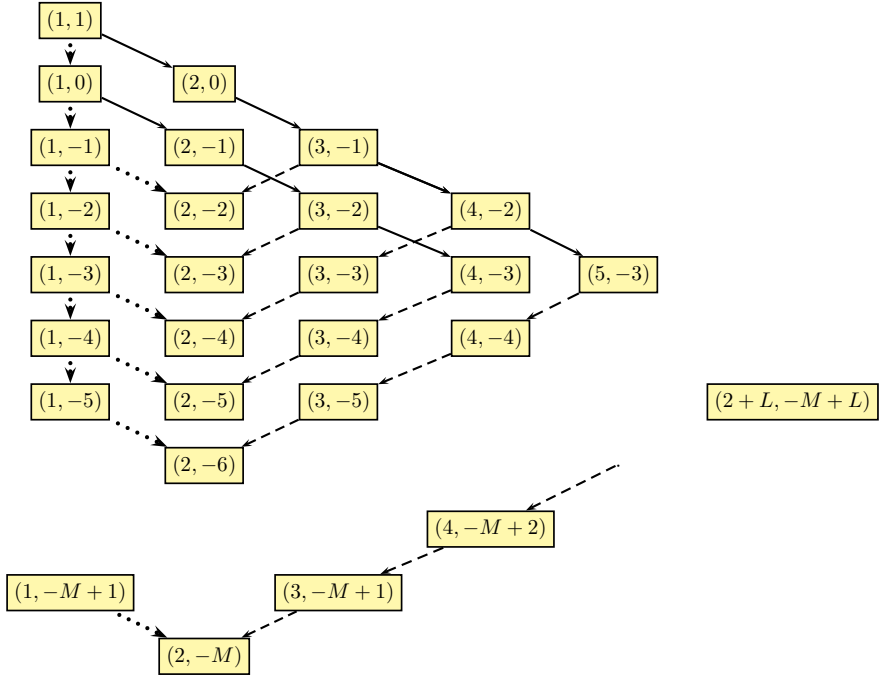

\newpsstyle{yellow}{fillstyle=solid,fillcolor=yellow!40}

\centering
\psscalebox{0.8}{%
\begin{psmatrix}[rowsep=0.4cm, colsep=0.4cm]
\psframebox[style=yellow]{$(1,1)$} & & & & & & & \\
\psframebox[style=yellow]{$(1,0)$} & \psframebox[style=yellow]{$(2,0)$} & & & & & & \\
\psframebox[style=yellow]{$(1,-1)$} & \psframebox[style=yellow]{$(2,-1)$} & \psframebox[style=yellow]{$(3,-1)$} & & & & & \\
\psframebox[style=yellow]{$(1,-2)$} & \psframebox[style=yellow]{$(2,-2)$} & \psframebox[style=yellow]{$(3,-2)$} & \psframebox[style=yellow]{$(4,-2)$} & & & & \\
\psframebox[style=yellow]{$(1,-3)$} & \psframebox[style=yellow]{$(2,-3)$} & \psframebox[style=yellow]{$(3,-3)$} & \psframebox[style=yellow]{$(4,-3)$} & \psframebox[style=yellow]{$(5,-3)$} & & & \\
\psframebox[style=yellow]{$(1,-4)$} & \psframebox[style=yellow]{$(2,-4)$} & \psframebox[style=yellow]{$(3,-4)$} & \psframebox[style=yellow]{$(4,-4)$} & & & & \\
\psframebox[style=yellow]{$(1,-5)$} & \psframebox[style=yellow]{$(2,-5)$} & \psframebox[style=yellow]{$(3,-5)$} & & & & \psframebox[style=yellow]{$(2+L,-M+L)$} & \\
 & \psframebox[style=yellow]{$(2,-6)$} & & & $\cdot$ & & & \\
 & & & \psframebox[style=yellow]{$(4,-M+2)$} & & & & \\
\psframebox[style=yellow]{$(1,-M+1)$} & & \psframebox[style=yellow]{$(3,-M+1)$} & & & & & \\
 & \psframebox[style=yellow]{$(2,-M)$} & & & & & & \\
\psset{linewidth=1pt} 
\ncline{->}{1,1}{2,2}^{}
\ncline{->}{2,2}{3,3}^{}
\ncline{->}{3,3}{4,4}>{}
\ncline{->}{4,4}{5,5}^{}
\ncline{->}{3,3}{4,4}>{}
\ncline{->}{2,1}{3,2}^{}
\ncline{->}{3,2}{4,3}^{}
\ncline{->}{4,3}{5,4}>{}
\psset{linewidth=2pt,linestyle=dotted}
\ncline{->}{1,1}{2,1}>{}
\ncline{->}{2,1}{3,1}>{}
\ncline{->}{3,1}{4,1}^{}
\ncline{->}{3,1}{4,2}^{}
\ncline{->}{4,1}{5,1}^{}
\ncline{->}{4,1}{5,2}^{}
\ncline{->}{5,1}{6,1}>{}
\ncline{->}{5,1}{6,2}>{}
\ncline{->}{6,1}{7,1}^{}
\ncline{->}{6,1}{7,2}^{}
\ncline{->}{7,1}{8,2}^{}
\ncline{->}{10,1}{11,2}^{}
\psset{linewidth=1pt,linestyle=dashed}
\ncline{->}{3,3}{4,2}^{}
\ncline{->}{4,3}{5,2}>{}
\ncline{->}{4,4}{5,3}^{}
\ncline{->}{5,3}{6,2}^{}
\ncline{->}{5,4}{6,3}>{}
\ncline{->}{6,3}{7,2}^{}
\ncline{->}{5,5}{6,4}^{}
\ncline{->}{6,4}{7,3}>{}
\ncline{->}{7,3}{8,2}^{}
\ncline{->}{8,5}{9,4}>{}
\ncline{->}{9,4}{10,3}>{}
\ncline{->}{10,3}{11,2}>{}
\end{psmatrix}
}
\caption{Solution schema plan labeled by indices $ (l,m) $ where the solved variable is $ W^{(m)}_{l} $.}
\end{figure*}
\vfill\eject

The results of our calculations undertaken using the above solution scheme are given, in the case of $ W_2 $
to order $ N^{-9} $, in the following proposition.
\begin{proposition}\label{W2expand}
Let $ s_1, s_2 $ be the first two elementary symmetric functions of $ z_1, z_2 $.
Furthermore $ z_1, z_2 $ are strictly bounded away from the unit circle.
The two-point resolvent function $ W_{2} $ in the $ \infty,0 $ domain has the expansion as $ N\to \infty $
\begin{multline}
    W_{2}\left(\begin{array}{cc} \infty, & 0 \\z_{1}, & z_{2} \end{array}\right) = -\frac{4}{\kappa}\frac{s_{2}}{(z_{1}-z_{2})^2}
\\
    -4\frac{(\kappa-1)}{\kappa^2N}\frac{s_{1}s_{2}}{(z_{1}-z_{2})^3}
    -4\frac{(\kappa -1)^2}{\kappa^3 N^2}\frac{s_2(s_1^2+2 s_2)}{\left(z_1-z_2\right)^4} 
\notag
\end{multline}
\begin{equation}
    -4\frac{(\kappa-1)}{\kappa^4 N^3}\frac{s_1 s_2}{\left(z_1-z_2\right)^5}\left[ (\kappa -1)^2 s_1^2+2 \left(4 \kappa ^2-7 \kappa +4\right) s_2 \right]
\notag
\end{equation}
\begin{multline}
    -4\frac{(\kappa-1)^2}{\kappa^5 N^4} \frac{s_2}{\left(z_1-z_2\right)^6}
    \\ \times
           \left[ (\kappa -1)^2 s_1^4+2 \left(11 \kappa ^2-16 \kappa +11\right) s_2 s_1^2+4 \left(4 \kappa ^2-5 \kappa +4\right) s_2^2 \right]
\notag
\end{multline}
\begin{multline}
    -4\frac{(\kappa-1)}{\kappa^6 N^5}\frac{s_1s_2}{\left(z_1-z_2\right)^7}
           \left[ (\kappa -1)^4 s_1^4+2 \left(26 \kappa ^4-81 \kappa ^3+111 \kappa ^2-81 \kappa +26\right) s_2 s_1^2 \right.
           \\ \left.
           +4 \left(34 \kappa ^4-95 \kappa ^3+126 \kappa ^2-95 \kappa +34\right) s_2^2 \right]
\notag
\end{multline}
\begin{multline}
    -4\frac{(\kappa-1)^2}{\kappa^7 N^6} \frac{s_2}{\left(z_1-z_2\right)^8}
           \left[ (\kappa-1)^4 s_1^6+6 \left(19 \kappa ^4-52 \kappa ^3+69 \kappa ^2-52 \kappa +19\right) s_2 s_1^4 \right.
           \\
           +72 \left(10 \kappa ^4-23 \kappa ^3+30 \kappa ^2-23 \kappa +10\right) s_2^2 s_1^2
           \\ \left.
           +4 \left(68 \kappa ^4-140 \kappa ^3+183 \kappa ^2-140 \kappa +68\right) s_2^3 \right]
\notag
\end{multline}
\begin{multline}
     -4\frac{(\kappa-1)}{\kappa^8 N^7} \frac{s_1 s_2 }{\left(z_1-z_2\right)^9}
            \left[ (\kappa -1)^6 s_1^6 \right.
            \\
            +2 \left(120 \kappa ^6-519 \kappa ^5+1044 \kappa ^4-1289 \kappa ^3+1044 \kappa ^2-519 \kappa +120\right) s_2 s_1^4 
            \\
            +8 \left(384 \kappa ^6-1449 \kappa ^5+2688 \kappa ^4-3233 \kappa ^3+2688 \kappa ^2-1449 \kappa +384\right) s_2^2 s_1^2
            \\ \left.
            +8 \left(496 \kappa ^6-1722 \kappa ^5+3051 \kappa ^4-3616 \kappa ^3+3051 \kappa ^2-1722 \kappa +496\right) s_2^3 \right]
\notag
\end{multline}
\begin{multline}
     -4\frac{(\kappa-1)^2}{\kappa^9 N^8}\frac{s_2}{\left(z_1-z_2\right)^{10}}
            \left[ (\kappa -1)^6 s_1^8 \right.
            \\
            +2 \left(247 \kappa ^6-960 \kappa ^5+1815 \kappa ^4-2192 \kappa ^3+1815 \kappa ^2-960 \kappa +247\right) s_2 s_1^6 
            \\
            +12 \left(968 \kappa ^6-3117 \kappa ^5+5390 \kappa ^4-6311 \kappa ^3+5390 \kappa ^2-3117 \kappa +968\right) s_2^2 s_1^4
            \\
            +8 \left(4288 \kappa ^6-12264 \kappa ^5+20319 \kappa ^4-23309 \kappa ^3+20319 \kappa ^2-12264 \kappa +4288\right) s_2^3 s_1^2
            \\ \left.
            +8 \left(992 \kappa ^6-2604 \kappa ^5+4212 \kappa ^4-4753 \kappa ^3+4212 \kappa ^2-2604 \kappa +992\right) s_2^4 \right]
\notag
\end{multline}
\begin{multline}
     -4\frac{(\kappa-1)}{\kappa^{10} N^9}\frac{s_1s_2}{\left(z_1-z_2\right)^{11}}
            \left[ (\kappa-1)^8 s_1^8 \right. 
            \\
            +2 \left(502 \kappa ^8-2725 \kappa ^7+7009 \kappa ^6-11461 \kappa ^5+13351 \kappa ^4 \right.
            \\ \left.
            -11461 \kappa ^3+7009 \kappa ^2-2725 \kappa +502\right) s_2 s_1^6
            \\
            +12 \left(3398 \kappa ^8-15783 \kappa ^7+36212 \kappa ^6-55308 \kappa ^5+63002 \kappa ^4 \right.
            \\ \left.
            -55308 \kappa ^3+36212 \kappa ^2-15783 \kappa +3398\right) s_2^2 s_1^4
            \\
            +16 \left(14384 \kappa ^8-60814 \kappa ^7+130739 \kappa ^6-192346 \kappa ^5+216458 \kappa ^4 \right.
            \\ \left.
            -192346 \kappa ^3+130739 \kappa ^2-60814 \kappa +14384\right) s_2^3 s_1^2
            \\
            +16 \left(11056 \kappa ^8-43750 \kappa ^7+90025 \kappa ^6-129211 \kappa ^5+144256 \kappa ^4 \right.
            \\ \left.\left.
            -129211 \kappa ^3+90025 \kappa ^2-43750 \kappa +11056\right) s_2^4 \right] .
\label{W2Ninfty}
\end{multline}
\end{proposition}

\begin{corollary}\label{momentLargeN}
The moments $ m_{k} $, $ k\geq 0 $ with $ k={\rm o}(N) $, possess the large $ N\to \infty $ expansion  
\begin{multline}
   m_{k}(N,\kappa)+N = \frac{1}{\kappa}k+\frac{(\kappa-1)}{N \kappa^2}k^2+\frac{(\kappa-1)^2}{N^2 \kappa^3}k^3
\\
   +\frac{(\kappa-1)}{6 N^3 \kappa^4}k^2\big[ -\kappa+\left(6 \kappa^2-11 \kappa +6\right) k^2 \big]
   +\frac{(\kappa-1)^2}{2 N^4 \kappa^5}k^3\big[ -\kappa+\left(2 \kappa^2-3 \kappa +2\right) k^2 \big]
\notag
\end{multline}
\begin{multline}
   +\frac{(\kappa-1)}{30 N^5 \kappa^6}k^2
               \big[ \kappa^3+\kappa^2+\kappa +\left(-30 \kappa^3+55 \kappa^2-30 \kappa \right) k^2
           \\
              +\left(30 \kappa^4-91 \kappa^3+124 \kappa^2-91 \kappa +30\right) k^4 \big]         
\notag
\end{multline}
\begin{multline}
   +\frac{ (\kappa-1)^2}{60 N^6 \kappa^7} k^3
               \big[ 8 \kappa^3+15 \kappa^2+8 \kappa +\left(-100 \kappa^3+150 \kappa^2-100 \kappa \right) k^2
           \\
              +\left(60 \kappa^4-148 \kappa^3+195 \kappa^2-148 \kappa +60\right) k^4 \big]
\notag
\end{multline}
\begin{multline}
   +\frac{(\kappa-1)}{840 N^7 \kappa^8}k^2
               \big[ -20 \kappa \left(\kappa^4+\kappa^3+\kappa^2+\kappa +1\right)
           \\
              +7 \kappa  \left(42 \kappa^4+31 \kappa^3-116 \kappa^2+31 \kappa +42\right) k^2
           \\ 
              -70 \kappa  \left(30 \kappa^4-91 \kappa^3+124 \kappa^2-91 \kappa +30\right) k^4
           \\ 
              +\left(840 \kappa^6-3214 \kappa^5+6033 \kappa^4-7288 \kappa^3 \right.
               \left.      +6033 \kappa^2-3214 \kappa +840\right) k^6 \big]
\notag
\end{multline}
\begin{multline}
   +\frac{(\kappa-1)^2}{5040 N^8 \kappa^9}k^3
              \big[ -600 \kappa^5-1112 \kappa^4-1180 \kappa^3-1112 \kappa^2-600 \kappa
            \\
              +\left(3780 \kappa^5+6048 \kappa^4-10605 \kappa^3+6048 \kappa^2+3780 \kappa \right) k^2
            \\
              +\left(-17640 \kappa^5+43512 \kappa^4-57330 \kappa^3+43512 \kappa^2-17640 \kappa \right) k^4
            \\
              +\left(5040 \kappa^6-15780 \kappa^5+27152 \kappa^4-31685 \kappa^3 \right.
                 \left. +27152 \kappa^2-15780 \kappa +5040\right) k^6 \big]
\notag
\end{multline}
\begin{multline}
   +\frac{(\kappa-1)}{7560 N^9 \kappa^{10}}k^2
              \big[  252 \kappa  \left(\kappa ^6+\kappa ^5+\kappa ^4+\kappa ^3+\kappa ^2+\kappa +1\right)
            \\
              -2 \kappa  \left(1470 \kappa ^6+1049 \kappa ^5-888 \kappa ^4-1162 \kappa ^3 \right.
              \left. -888 \kappa ^2+1049 \kappa +1470\right) k^2
            \\
              +21 \kappa  \left(510 \kappa ^6+235 \kappa ^5-2937 \kappa ^4+4552 \kappa ^3 \right.
              \left. -2937 \kappa ^2+235 \kappa +510\right) k^4
            \\
              -42 \kappa  \left(840 \kappa ^6-3214 \kappa ^5+6033 \kappa ^4-7288 \kappa ^3 \right.
              \left. +6033 \kappa ^2-3214 \kappa +840\right) k^6
            \\
              +\left(7560 \kappa ^8-33222 \kappa ^7+73603 \kappa ^6-110325 \kappa ^5+124936 \kappa ^4 \right.
            \\
              \left. -110325 \kappa ^3+73603 \kappa ^2-33222 \kappa +7560\right) k^8 \big] .
\label{mNinfty}
\end{multline}
\end{corollary}
\begin{proof}
Using a partial fraction decomposition of \eqref{W2Ninfty} with respect to $ \zeta=z_2/z_1 $ and then making the
substitutions $ (\zeta-1)^{-m}\mapsto (-1)^m\frac{(m)_{k}}{k!} $ for $ m=2,3,\ldots $ we directly deduce \eqref{mNinfty}.
\end{proof}

\begin{remark}
If one examines the coefficient of the highest $ k $ monomial in each of the expansion terms of \eqref{mNinfty},
as per the scaling \eqref{CLimit}, then one recovers the sequence of $ \kappa $ polynomials first reported in
Eq. (8.1) of \cite{FJM_2001} (recall \eqref{Sa})
\begin{gather*}
 1-\frac{11 \kappa }{6}+\kappa^2,
 \quad
 1-\frac{3 \kappa }{2}+\kappa^2,
 \\
 1-\frac{91 \kappa }{30}+\frac{62 \kappa^2}{15}-\frac{91 \kappa^3}{30}+\kappa^4,
 \quad
 1-\frac{37 \kappa }{15}+\frac{13 \kappa^2}{4}-\frac{37 \kappa^3}{15}+\kappa^4,
 \\
 1-\frac{1607 \kappa }{420}+\frac{2011 \kappa^2}{280}-\frac{911 \kappa^3}{105}+\frac{2011 \kappa^4}{280}-\frac{1607 \kappa^5}{420}+\kappa^6,
 \\
 1-\frac{263 \kappa }{84}+\frac{1697 \kappa^2}{315}-\frac{6337 \kappa^3}{1008}+\frac{1697 \kappa^4}{315}-\frac{263 \kappa^5}{84}+\kappa^6,
 \\
 1-\frac{791 \kappa }{180}+\frac{73603 \kappa^2}{7560}-\frac{7355 \kappa^3}{504}+\frac{2231 \kappa^4}{135}-\frac{7355 \kappa^5}{504}+\frac{73603 \kappa^6}{7560}-\frac{791 \kappa^7}{180}+\kappa^8,
\end{gather*}
where in their work we identify $ x\mapsto \kappa $, $ y\mapsto 1/N\kappa $.
\end{remark}

\begin{proposition}\label{mRational}
The low index moments have the exact rational evaluations
\begin{gather}
   m_{0}(N,\kappa) = -N ,
\label{0th-m} \\
   m_{1}(N,\kappa) = -N+\frac{1}{\kappa}+\frac{(\kappa-1)}{\kappa(\kappa N+1-\kappa)} = -N+\frac{N}{\kappa N+ 1-\kappa} ,
\label{1st-m} \\
   m_{2}(N,\kappa) = -N+\frac{2}{\kappa}
\label{2nd-m} \\
        +\frac{(\kappa-1)}{\kappa}\left[ \frac{2}{\kappa N+1-\kappa}-\frac{2(\kappa-2)}{(\kappa+1)(\kappa N+2-\kappa)}+\frac{2(2\kappa-1)}{(\kappa+1)(\kappa N+1-2\kappa)} \right] .
\notag
\end{gather}
\end{proposition}
\begin{proof}
Given that we have at hand $ 10 $ orders in the expansion of $ m_k $ in \eqref{mNinfty} we investigate a 
$ [j;j+1] $ Pad\'e analysis of the low index $ k $ examples of $ m_k+N-k/\kappa $ with respect to $ N $ 
about $ N=\infty $. The reason for this type of
Pad\'e approximant is that $ {\rm deg}_N(den) = {\rm deg}_N(num)+1 $. For $ m_1+N-1/\kappa $ we find the
$ [1;2] $ approximant yields a rational function of $ N $ which agrees with all terms in the expansion \eqref{mNinfty}.
Another signature of this fit is that higher approximants yield an indeterminate situation, i.e. vanishing of
all subsequent Hankel determinants. This is \eqref{1st-m}.
For $ m_2+N-2/\kappa $ we find the same situation in the case of the $ [3;4] $ approximant, and $ [4;5] $ and
higher approximants are indeterminate. The result is \eqref{2nd-m}. For $ k \ge 3 $ we expect a $ [5;6] $
approximant would be sufficient however we do not have enough terms in the expansion for this.

In case the reader may doubt the veracity of the formulae \eqref{0th-m}-\eqref{2nd-m} one can in fact directly
prove these claims. If one takes the second Loop Equation \eqref{LE:2} with $ z\in\compD $ and $ z_1\in\compD $
then terms with $ W_2(z,z_1) $, $ W_3(z,z,z_1) $ vanish and $ W_1(z)=W_1(z_1)=-N $, so that it reduces to
\begin{equation*}
    \frac{2N}{z_1}+\tfrac{1}{2}[\kappa(N-1)+1]\left.\partial_{z_0} W_2(z_0,z_1)\right|_{z_0=0} = 0 .
\end{equation*} 
However from \eqref{DysonW2} $ \left.\partial_{z_0} W_2(z_0,z_1)\right|_{z_0=0} = -4(m_1+N)/z_1 $ and we deduce
\eqref{1st-m}.
\end{proof}

\begin{remark}
As we will see in the next Section the formulae \eqref{0th-m}-\eqref{2nd-m} agree with $ \kappa=2 $ CSE result
\eqref{CSEparfrac} and the $ \kappa=1/2 $ COE result \eqref{COEparfrac} for all $ N $, and with the $ N=2 $
cases \eqref{evenN=2:1},  \eqref{evenN=2:2} and the $ N=3 $ cases \eqref{evenN=3:1}, \eqref{evenN=3:2} for all 
even, positive $ \beta $. One might speculate that the form for general $ k $ begins with
\begin{equation}
   m_{k}(N,\kappa) = -N+\frac{k}{\kappa}+\frac{k(\kappa-1)}{\kappa}\frac{1}{\kappa N+1-\kappa} + \ldots ? .
\label{kth-m} 
\end{equation} 
\end{remark}

In the context of the circular Dyson ensembles we observe the following duality formulae are valid.
\begin{proposition}\label{Dual}
The moments satisfy the duality relation, where both sides are non-zero and $ \kappa \neq 0,\infty $
\begin{equation}
  \kappa ^{-2}(m_k+N)(-\kappa N,\kappa^{-1}) = (m_k+N)(N,\kappa) .
\label{mDuality}
\end{equation} 
The resolvent functions in the global regime satisfy the duality relations $ \kappa \neq 0,\infty $, $ l\in \N $
\begin{equation}
   (-1)^{l}\kappa^{-l} W_{l}(z_{1},\ldots,z_{l},-\kappa N,\kappa^{-1}) = W_{l}(z_{1},\ldots,z_{l},N,\kappa) .
\label{WnDuality}
\end{equation} 
\end{proposition}

\begin{remark}
Analogous dualities for the moments in the Gaussian $ \beta $ ensemble were established using Jack polynomial
theory in the study of Dimitriu and Edelman \cite{DE_2006}, and the corresponding results for the generating
functions were given in \cite{WF_2014}.
\end{remark}

\section{Special Exact Cases of the Dyson Circular Ensembles}\label{specialExact}
\setcounter{equation}{0}

\subsection{General \texorpdfstring{$ N $}{N} and \texorpdfstring{$ \beta = 1, 2, 4 $}{CXE} Circular Ensembles}\label{beta124DysonEnsembles}
We recount and extend some of the well-known results for the two-point correlations in the $ \beta=1,2,4 $ cases 
for the purposes of comparison to and checking against the results found for general $ \beta $ in the preceding
section. These special cases also serve as illustrations of some key properties of the two-point resolvents
for general $ \beta $. At the same time we also highlight some of the differences between the exact results
and those found within the global expansion regime, which will arise from contributions to $ m_k(N,\kappa) $
when $ k={\rm O}(N) $ and a failure of analyticity.
We should point out that our formulation has differing normalisation conventions to
those of Mehta \cite{Meh_2004}, Chapter 10. Let us define \cite{Meh_2004}, Eq. (10.1.6) and (10.1.3)
\begin{equation}
  S_{N}(\theta) = \sum_{p\in \mathscr{A}_N}e^{ip\theta} ,
\label{Sdefn}
\end{equation} 
where $ \mathscr{A}_N = \{\frac{1}{2}(1-N),\frac{1}{2}(3-N),\ldots,\frac{1}{2}(N-3),\frac{1}{2}(N-1)\} $
and has the properties
$ S_{N}(-\theta) = S_{N}(\theta) $, $ S_{N}(\theta+2\pi) = (-1)^{N-1}S_{N}(\theta) $, $ S_{N}(0) = N $. 
Alternatively one has an evaluation in terms of the second Chebyshev polynomials
\begin{equation*}
  S_{N}(\theta) = \frac{\sin\frac{1}{2}N\theta}{\sin\frac{1}{2}\theta} = U_{N-1}(\cos\tfrac{1}{2}\theta) ,
\end{equation*}
In addition we require the definition of the angular derivative \cite{Meh_2004}, Eq. (10.3.6)
\begin{equation}
   DS_{N}(\theta) := \frac{d}{d\theta} S_{N}(\theta) ,
\label{DSdefn}
\end{equation}
which can be expressed in terms of the first and second Chebyshev polynomials
\begin{equation*}
   DS_{N}(\theta) = \tfrac{1}{2}N{\rm cosec}\tfrac{1}{2}\theta\, T_{N}(\cos\tfrac{1}{2}\theta) - \tfrac{1}{2}\cot\tfrac{1}{2}\theta\, U_{N-1}(\cos\tfrac{1}{2}\theta) .
\end{equation*} 
Furthermore we make the definition of the indefinite integral \cite{Meh_2004}, Eq. (10.3.7)
\begin{equation*}
    IS_{N}(\theta) := \int^{\theta}_{0}d\theta'\;S_{N}(\theta') ,
\end{equation*}
which has the trigonometric series representation
\begin{equation}
   IS_{N}(\theta) =
    \begin{cases}\displaystyle
     4\sum_{p=1,3,\ldots, N-1}\frac{\sin\frac{1}{2}p\theta}{p} , &  N \in 2\mathbb{Z} \\ 
                \displaystyle
     4\sum_{p=2,4,\ldots, N-1}\frac{\sin\frac{1}{2}p\theta}{p} + \theta , &  N \in 2\mathbb{Z}+1                                     
   \end{cases} .
\label{ISdefn}
\end{equation} 
A related quantity to the above, is \cite{Meh_2004}, Eq. (10.3.10) and (10.3.11)
\begin{equation*}
   JS_{N}(\theta) := -\frac{1}{i}\sum_{q\in \mathscr{R}_{N}} q^{-1}e^{iq\theta} ,
\end{equation*}
where the summation is unbounded over the index set $ \mathscr{R}_{N}=\{\pm\frac{1}{2}(N+1) $, $ \pm\frac{1}{2}(N+3),\ldots\} $. 
Similarly one has the alternative expression
\begin{equation}
   JS_{N}(\theta) = -4\sum_{p=N+1,N+3,\ldots}\frac{\sin\frac{1}{2}p\theta}{p} .
\label{JSdefn}
\end{equation}
Lastly we define $ \epsilon_{N}(\theta) := IS_{N}(\theta)-JS_{N}(\theta) $ which has either a ``saw-tooth'' profile
\begin{equation*}
   \epsilon_{N\in 2\mathbb{Z}}(\theta) =
    \begin{cases}\displaystyle
     (-1)^m\pi , &  2\pi m < \theta < 2\pi(m+1) \\ 
                \displaystyle
     0 , &  \theta = 2\pi m
   \end{cases}, \quad m \in \mathbb{Z} ,
\end{equation*}
or a ``step'' profile
\begin{equation*}
   \epsilon_{N\in 2\mathbb{Z}+1}(\theta) =
    \begin{cases}\displaystyle
     (2m+1)\pi , &  2\pi m < \theta < 2\pi(m+1) \\ 
                \displaystyle
     2m\pi , &  \theta = 2\pi m 
   \end{cases}, \quad m \in \mathbb{Z} .
\end{equation*}

\subsection{\texorpdfstring{$ \beta=2 $}{CUE} CUE}\label{CUE}
A standard result gives, see \cite{Meh_2004} pg. 196 Eq. (10.1.13), 
\begin{equation*}
 \rho_{(2)C}(\theta) = -\left( S_{N}(\theta) \right)^2 
                    = -\frac{(1-z^{N})(1-z^{-N})}{(1-z)(1-z^{-1})} .
\end{equation*}
From its Fourier decomposition the moments are
\begin{equation*}
   m_{k} = -(N-|k|)\Theta(N-|k|) , 
\end{equation*} 
and we note that this is not analytic at $ k=N $. One can readily see that this has an exact large $ N $ continuum limit
\begin{equation*}
   m(x) = -(1-|x|)\Theta(1-|x|) .  
\end{equation*} 
The two-point resolvent function is readily computed and found to be given by
\begin{equation*}
   W_{2}(z_{1},z_{2}) = 
   \begin{cases}\displaystyle
      -4\zeta\frac{1-\zeta^{N}}{(1-\zeta)^2}, &\displaystyle \zeta = \frac{z_{2}}{z_{1}}, \quad \text{$ z_{1}\in \compD $, $ z_{2}\in \D $ or vice versa} \\
      0, & \text{otherwise}
   \end{cases} .
\end{equation*} 
This differs from the leading, universal term of \eqref{W2Ninfty} by the term $ \zeta^N $ in the numerator, which
is not accessible in the global regime. It is interesting to observe that $ W_2(\zeta) $ satisfies the Bieberbach
property $ |m_k+N|\leq |k| $ for all $ k, N $.

\subsection{\texorpdfstring{$ \beta=4 $}{CSE} CSE}\label{CSE}
Again from \cite{Meh_2004}, pg. 211 Eq. (10.5.6) and pg. 212 Eq. (10.5.15), we have
\begin{equation*}
   \rho_{(2)C}(\theta) = -\frac{1}{4}\left[ (S_{2N}(\theta))^2-DS_{2N}(\theta)IS_{2N}(\theta) \right] . 
\end{equation*}
\begin{proposition}\label{CSEmoment}
The moments of the two-point resolvent function for the CSE are given by $ |k| \leq 2N-2 $
\begin{equation*}
   m_{k} = -\tfrac{1}{2}
    \begin{cases}\displaystyle
     2N-k-\tfrac{1}{2}k\big[ \psi(N+\tfrac{1}{2})-\psi(-N+k+\tfrac{1}{2}) \big] , &  k > 0 \\ 
                \displaystyle
     2N+k-\tfrac{1}{2}k\big[ \psi(N+k+\tfrac{1}{2})-\psi(-N+\tfrac{1}{2}) \big] , &  k \leq 0 \\ 
   \end{cases} .  
\end{equation*} 
For $ |k| \geq 2N-1 $, $ m_{k}=0 $. Here $ \psi(x) $ is the standard di-gamma function, see Eq. 5.2.2,
{\tt http://dlmf.nist.gov/5.2.E2} in \cite{DLMF}. 
\end{proposition}
\begin{proof}
Employing the definitions \eqref{Sdefn} and \eqref{DSdefn} we find 
\begin{equation*}
   m_{k} = -\tfrac{1}{2} \sum^{2N-1+\min(0,k)}_{l=\max(0,k)} \frac{2l-2N+1-k}{2l-2N+1} ,
\end{equation*}
or without loss of generality the partial fraction sum formula valid for $ k>0 $
\begin{equation}
   m_{k} = -N+\frac{k}{2}+\frac{k}{2} \left[ \frac{1}{2N-1}+\frac{1}{2N-3}+\ldots+\frac{1}{2N-(2k-1)} \right] .
\label{CSEparfrac}
\end{equation}
\end{proof}
There are a couple of observations to note here - as well as terminating at $ |k|=2N-2 $, $ m_k+N $ has a maximum
at $ k=N $ of $  \frac{1}{2}N+\frac{1}{4}N\left[ \psi(N+\frac{1}{2})-\psi(\frac{1}{2}) \right] $. Thus $ W_2 $,
in this case, violates the Bieberbach inequality and fails to be univalent. The large $ N $ continuum limit is
\begin{equation*}
   m(x) = -1+\tfrac{1}{2}x-\tfrac{1}{4}x\log|1-x| ,
\end{equation*} 
which exhibits a weak non-analyticity at $ x=1 $.

\begin{proposition}
The second resolvent function for the CSE in the $ 0,\infty $ domain is given by $ \zeta=z_{1}/z_{2}<1 $
\begin{multline}
   W_2(\zeta) = 2(1-\zeta)^{-1}(1-\zeta^{2N-1})
\\
      +(1-\zeta)^{-2}\left[ \zeta^2(1-\zeta^{2N-2})-2(1-\zeta^{2N-1})-\frac{2}{2N-1}\zeta(1-\zeta^{2N-1}) \right]
\\
         -\left[ 2\zeta(1-\zeta)^{-2}+(2N+1)(1-\zeta)^{-1} \right]
\\ \times
          \left[ \frac{\zeta^2}{2N-3}{}_2F_1(1,\tfrac{3}{2}-N;\tfrac{5}{2}-N;\zeta)+\frac{\zeta^{2N}}{2N-1}{}_2F_1(1,N-\tfrac{1}{2};N+\tfrac{1}{2};\zeta) \right] .  
\label{CSEexact}
\end{multline}
In the global asymptotic regime we have as $ N\to \infty $
\begin{multline}
  W_2(\zeta) \sim -\frac{2 \zeta}{(\zeta-1)^2}+\frac{1}{N(\zeta-1)^3}\zeta \left(1+\zeta\right)
 -\frac{1}{2N^{2}(\zeta-1)^{4}}\zeta \left(1+4 \zeta+\zeta^2\right)
\\
 +\frac{1}{4N^{3}(\zeta-1)^{5}}\zeta \left(1+15 \zeta+15 \zeta^2+\zeta^3\right)
\\
 -\frac{1}{8N^{4}(\zeta-1)^{6}}\zeta \left(1+50 \zeta+138 \zeta^2+50 \zeta^3+\zeta^4\right)
\\
 +\frac{1}{16N^{5}(\zeta-1)^{7}}\zeta \left(1+157 \zeta+994 \zeta^2+994 \zeta^3+157 \zeta^4+\zeta^5\right)
\\
 -\frac{1}{32N^{6}(\zeta-1)^{8}}\zeta \left(1+480 \zeta+6231 \zeta^2+13456 \zeta^3+6231 \zeta^4+480 \zeta^5+\zeta^6\right)
\\
 +\frac{1}{64N^{7}(\zeta-1)^{9}}\zeta \left(1+1451 \zeta+35961 \zeta^2+146907 \zeta^3+146907 \zeta^4 \right. 
  \\                                  \left. +35961 \zeta^5+1451 \zeta^6+\zeta^7\right)
\\
 -\frac{1}{128N^{8}(\zeta-1)^{10}}\zeta \left(1+4366 \zeta+197224 \zeta^2+1402834 \zeta^3+2597230 \zeta^4 \right. 
  \\                                  \left. +1402834 \zeta^5+197224 \zeta^6+4366 \zeta^7+\zeta^8\right)
\\
 +\frac{1}{256N^{9}(\zeta-1)^{11}}\zeta \left(1+13113 \zeta+1047252 \zeta^2+12262436 \zeta^3+38286798 \zeta^4 \right. 
  \\                                  \left. +38286798 \zeta^5+12262436 \zeta^6+1047252 \zeta^7+13113 \zeta^8+\zeta^9\right) .
\label{CSENinfty}
\end{multline}
\end{proposition}
\begin{proof}
For details we refer the reader to the proof of Proposition \ref{COE_W2} as entirely identical methods apply to both cases.
\end{proof}

\subsection{\texorpdfstring{$ \beta=1 $}{COE} COE}\label{COE}
From \cite{Meh_2004}, pg. 201 Eq. (10.3.16) and pg. 205 Eq. (10.3.42), we have
\begin{equation*}
   \rho_{(2)C}(\theta) = -\left[ (S_{N}(\theta))^2-DS_{N}(\theta)JS_{N}(\theta) \right] .
\end{equation*}
\begin{proposition}\label{COEmoment}
The moments of the COE two-point resolvent function are given by
\begin{multline*}
   m_{k} = -(N-|k|)\Theta(N-|k|)- 
\\
    \begin{cases}\displaystyle
     -\min(k,N) \\ +k\big[ \psi(\tfrac{1}{2}(N+1)+k)-\psi(\tfrac{1}{2}(N+1)+\max(0,k-N)) \big] , &  k > 1 \\ 
                \displaystyle
     -\min(-k,N) \\ -k\big[ \psi(\tfrac{1}{2}(N+1)-k)-\psi(\tfrac{1}{2}(N+1)+\max(0,-k-N)) \big] , &  k \leq -1 \\ 
   \end{cases} . 
\end{multline*}
Note that $ m_k $ is non-zero for all $ k $.
\end{proposition}
\begin{proof}
In this case we need to employ the formulae \eqref{DSdefn} and \eqref{JSdefn} from which we compute
\begin{multline*}
   m_{k} = -\left\{ N-|k|+\sum^{k-1}_{l=\max(0,k-N)}\frac{k-1-l+\frac{1}{2}(1-N)}{l+\frac{1}{2}(N+1)} \right.
\\ \left.
                         -\sum^{-k-1}_{l=\max(0,-k-N)}\frac{k+N+l+\frac{1}{2}(1-N)}{l+\frac{1}{2}(N+1)}
            \right\} ,
\end{multline*}
or in form of the partial-fraction sum, when $ k>0 $
\begin{equation}
  m_{k} = -N+2k-2k\left[ \frac{1}{N+1}+\frac{1}{N+3}+\ldots+\frac{1}{N+2k-1} \right] .
\label{COEparfrac}
\end{equation}
\end{proof}
The moment $ m_k $ does not have a maximum for finite $k$ but approaches zero as $ k\to\infty $.
The large $ N $ continuum limit of $ m_k $ is now
\begin{equation*}
  m(x) = 
 \begin{cases}
    -1+2x-x\log(2x+1), & 0<x<1 \\ 
    1-x \log\displaystyle\frac{2x+1}{2x-1}, & x>1 
 \end{cases} ,
\end{equation*}
and is not analytic at $ x=1 $ (very weakly though, as the difference between either side of $ x=1 $ first appears
at the third order).

\begin{proposition}\label{COE_W2}
The second resolvent function for the COE in the $ 0,\infty $ domain is given by $ \zeta=z_{1}/z_{2}<1 $
\begin{multline}
   W_2(\zeta) = 4(1-\zeta)^{-1}\zeta^N+4(1-\zeta)^{-2}\left[ -2\zeta(1-\zeta^N)+\zeta-\zeta^N \right]
\\
         +\frac{4}{N+1}\left[ 2\zeta^2(1-\zeta)^{-2}(1-\zeta^{N})-\zeta(1-\zeta)^{-1}(N-1+(N+1)\zeta^N) \right]
\\ \times
          {}_2F_1(1,\tfrac{1}{2}N+\tfrac{1}{2};\tfrac{1}{2}N+\tfrac{3}{2};\zeta) .        
\label{COEexact}
\end{multline}
In the global asymptotic regime we have as $ N\to \infty $
\begin{multline}
   W_2(\zeta) \sim -\frac{8 \zeta}{(\zeta-1)^2}-\frac{8}{N (\zeta-1)^3} \zeta (1+\zeta)-\frac{8}{N^2 (\zeta-1)^4} \zeta \left(1+4 \zeta+\zeta^2\right)
\\
 -\frac{8}{N^3 (\zeta-1)^5} \zeta \left(1+15 \zeta+15 \zeta^2+\zeta^3\right)
\\
 -\frac{8}{N^4 (\zeta-1)^6} \zeta \left(1+50 \zeta+138 \zeta^2+50 \zeta^3+\zeta^4\right)
\\
 -\frac{8}{N^5 (\zeta-1)^7} \zeta \left(1+157 \zeta+994\zeta^2+994\zeta^3+157\zeta^4+\zeta^5\right)
\\
 -\frac{8}{N^6 (\zeta-1)^8} \zeta \left(1+480 \zeta+6231 \zeta^2+13456 \zeta^3+6231 \zeta^4+480 \zeta^5+\zeta^6\right)
\\
 -\frac{8}{N^7 (\zeta-1)^9} \zeta \left(1+1451 \zeta+35961 \zeta^2+146907 \zeta^3 \right. 
  \\
                                  \left. +146907 \zeta^4+35961 \zeta^5+1451 \zeta^6+\zeta^7\right)
\\
 -\frac{8}{N^8 (\zeta-1)^{10}} \zeta \left(1+4366 \zeta+197224 \zeta^2+1402834 \zeta^3+2597230 \zeta^4 \right. 
  \\
                                  \left. +1402834 \zeta^5+197224 \zeta^6+4366 \zeta^7+\zeta^8\right)
\\
 -\frac{8}{N^9 (\zeta-1)^{11}} \zeta \left(1+13113 \zeta+1047252 \zeta^2+12262436 \zeta^3+38286798 \zeta^4 \right. 
  \\
                                  \left. +38286798 \zeta^5+12262436 \zeta^6+1047252 \zeta^7+13113 \zeta^8+\zeta^9\right) .
\label{COENinfty}
\end{multline}
\end{proposition}
\begin{proof}
A rather tedious exercise left for the reader. In the simplification of the Gau{\ss} hypergeometric functions 
we have employed the identity
\begin{equation*}
    {}_2F_1(1,b+1;c+1;z) = \frac{c}{bz}\left[ {}_2F_1(1,b;c;z)-1 \right] ,
\end{equation*} 
which is valid for $ b, z \neq 0 $. In addition we have the special case $ c=b+1 $. For the global expansions
we have used the identity
\begin{equation*}
     {}_2F_1(a,b+\lambda;c+\lambda;z) = (1-z)^{-a}{}_2F_1(a,c-b;c+\lambda;z(z-1)^{-1}) ,
\end{equation*}
and expanded the resulting Gau{\ss} hypergeometric functions term-wise.
\end{proof}

\begin{remark}
The direct evaluations of the global expansions of $ W_2 $ for $ \beta=4,1 $, namely \eqref{CSENinfty} and \eqref{COENinfty}
respectively, agree with the appropriate specialisations of \eqref{W2Ninfty}.
In respect of the moments we can readily verify from \eqref{COEparfrac} and \eqref{CSEparfrac} that they satisfy
\begin{equation}
   (m_{k}+N)(-2N,\tfrac{1}{2}) = 4(m_{k}+N)(N,2) ,\quad |k| \leq 2N-2 ,
\label{auxDual}
\end{equation} 
and that in the global regime we have $ W_2(\zeta;-2N,\tfrac{1}{2})=4W_2(\zeta;N,2) $ as is evident by comparing
\eqref{CSENinfty} and \eqref{COENinfty}. Thus there is consistency with Prop. \ref{Dual}.
However the exact forms \eqref{CSEexact} and \eqref{COEexact} do not satisfy this latter
relation because the symplectic moments terminate whilst the orthogonal ones do not even though the first set of
$ 2N-2 $ moments are related by \eqref{auxDual}.
\end{remark}

\begin{remark}
In computing the large $ N $ expansions of the resolvent functions $ W_2 $ for these special $ \beta $'s 
we employed explicit and elementary function representations of the corresponding densities, without the
need of other methods and in particular the use of skew-orthogonal polynomials. The asymptotics of skew-orthogonal polynomials
has been studied in \cite{Eyn_2001a} where one can find the leading order asymptotics for the skew-orthogonal polynomials
for a polynomial potential and then applied to the kernels and the two-point correlations, for $ \beta=1,2,4 $.
However we have the exact forms from which the large $ N $ expansions are readily and systematically constructed to 
any order of approximation.
\end{remark}

\subsection{Even \texorpdfstring{$ \beta\in 2\mathbb{Z} $}{b2Z} and Small \texorpdfstring{$ N $}{N} Dyson Circular Ensembles}\label{CorrelationDuality}

Further insight can be provided by the special cases of $ \beta\in 2\mathbb{Z} $ through the duality property.
The duality formula of Prop. (13.2.2) p. 603 \cite{For_2010}, or Eqs. (3.9), (3.13) and (3.14) of \cite{For_1995}
for the unconnected two-point correlation states 
\begin{multline}
  \rho_{(2)}(\theta;N) = N(N-1)\frac{\Gamma(\kappa(N+1)+1)\Gamma(\kappa+1)^{3}}{\Gamma(\kappa(N-1)+1)\Gamma(3\kappa+1)\Gamma(2\kappa+1)}
\\  \times
          \prod^{\beta}_{j=1} \frac{\Gamma(2+j\kappa^{-1})\Gamma(1+\kappa^{-1})}{\Gamma(j\kappa^{-1})^2\Gamma(1+j\kappa^{-1})}
    |e^{i\theta}-1|^{2\kappa}e^{-i\kappa(N-2)\theta} \times
\\ 
    \int_{[0,1]^{\beta}} dx_1 \ldots dx_{\beta} \prod^{\beta}_{j=1}\left[ x_j(1-x_j) \right]^{\frac{\scriptstyle 1}{\scriptstyle \kappa}-1}\left[ 1-(1-e^{i\theta})x_j \right]^{N-2}
                                                \prod_{1\leq j<k \leq \beta}|x_j-x_k|^{\frac{\scriptstyle 2}{\scriptstyle \kappa}} .
\label{dualIntegral}
\end{multline}

Evaluating \eqref{dualIntegral} when $ N=2 $ is just an instance of the Selberg integral, see Eqs. (4.1) and (4.3)
of \cite{For_2010} and yields
\begin{equation*}
    \rho_{(2)C}(\theta;2) = 2\frac{\Gamma(\kappa+1)\Gamma(1/2)}{\Gamma(\kappa+1/2)}\left|\sin\tfrac{1}{2}\theta\right|^{\beta}-4 .
\end{equation*}
To compute the Fourier decomposition of this density we require the integral, valid for $ {\rm Re}\;\beta>-1 $, $ k\in \Z $
\begin{equation}
   \int^{2\pi}_{0}\frac{d\theta}{2\pi} e^{i k\theta}\left( \sin\tfrac{1}{2}\theta \right)^{\beta} = 
   (-1)^k2^{-\beta}\frac{\Gamma(1+2\kappa)}{\Gamma(1+\kappa+k)\Gamma(1+\kappa-k)} ,
\label{Fintegral}
\end{equation}
and from this we read off 
\begin{equation*}
  m_{k} = -4\delta_{k,0}+2(-1)^k\frac{\Gamma^2(1+\kappa)}{\Gamma(1+k+\kappa)\Gamma(1-k+\kappa)} .
\end{equation*}
Some low index examples, for which one can make comparisons with our other results, are
\begin{align}
   m_{1}+2 & = \frac{2}{1+\kappa} ,
\label{evenN=2:1}\\
   m_{2}+2 & = 4+\frac{4}{1+\kappa}-\frac{12}{2+\kappa} ,
\label{evenN=2:2}\\
   m_{3}+2 & = \frac{6}{1+\kappa}-\frac{48}{2+\kappa}+\frac{60}{3+\kappa} ,
\label{evenN=2:3}\\
   m_{4}+2 & = 4+\frac{8}{1+\kappa}-\frac{120}{2+\kappa}+\frac{360}{3+\kappa}-\frac{280}{4+\kappa} .
\label{evenN=2:4}
\end{align}

An evaluation of \eqref{dualIntegral} is also possible for $ N=3 $, however this task is a little more involved.

\begin{lemma}
For $ N=3 $ and $ \beta\in 2\N $ the connected two-point correlation function is
\begin{multline}
    \rho_{(2)C}(\theta;3) = -9
\\
    +6\frac{\Gamma(\kappa+1)^3\Gamma(4\kappa+1)}{\Gamma(3\kappa+1)\Gamma(2\kappa+1)^2}
                           e^{-\frac{1}{2}i\beta\theta}\left|2\sin\tfrac{1}{2}\theta\right|^{\beta}{}_{2}F_{1}(-\beta,-\beta;-2\beta;1-e^{i\theta}) .
\label{rhoN=3}                           
\end{multline}
\end{lemma}
\begin{proof}
For this case we expand the additional factor $ \prod^{\beta}_{j=1}\left[ 1-(1-e^{i\theta})x_j \right] $ in
\eqref{dualIntegral} as a
polynomial in $ (1-e^{i\theta}) $ with elementary symmetric function coefficients. Aomoto's extension of the Selberg
integral allows us to calculate this via a recurrence relation (see Eq. (4.130) of \cite{For_2010}), and thus with
\begin{equation*}
   I^{\alpha=1}_m := 
    \int_{[0,1]^{\beta}} dx_1 \ldots dx_{\beta} \prod^{\beta}_{j=1}\left[ x_j(1-x_j) \right]^{\kappa^{-1}-1}x_1 \cdots x_m
                                                \prod_{1\leq j<k \leq \beta}|x_j-x_k|^{2\kappa^{-1}} ,
\end{equation*}
we find $ I^{\alpha=1}_m=\frac{\displaystyle (-\beta)_m}{\displaystyle (-2\beta)_m}I^{\alpha=1}_0 $
where $ I^{\alpha=1}_0 $ is the standard Selberg integral. Applying this evaluation into the polynomial and resumming
we deduce the result \eqref{rhoN=3}. 
\end{proof}

\begin{proposition}
In the $ \beta\in2\N $ case with $ N=3 $ the moments are given by 
\begin{multline}
   m_k = 
   -9\delta_{k,0}+ 6(-1)^k\cos\pi\kappa\frac{\Gamma(\kappa+1)^3\Gamma(4\kappa +1)}{\Gamma(3\kappa+1)\Gamma(2\kappa +1)}
\\ \times
                      \frac{1}{\Gamma(1+2\kappa-k)\Gamma(1+k)}{}_{3}F_{2}(-\beta ,-\beta ,1+\beta;-2\beta ,k+1;1) .
\label{momN=3}                      
\end{multline}
Some care needs to be exercised in interpreting this hypergeometric function with negative integer numerator parameters,
and here we simply mean the terminating sum implied by the first parameter with $ \beta=2M $, $ M\in\N $.
Alternatively this can be expressed as a $ 2k $-sum
\begin{multline}
   m_k = -9\delta_{k,0}+6(-1)^k\cos\pi\kappa \frac{2^{4\kappa}\Gamma(\kappa+1)^3}{\Gamma(3\kappa+1)\Gamma(2\kappa +1+k)\Gamma(2\kappa +1-k)}
\\ \times
   \sum^{2k}_{i=0}(-1)^i\binom{2k}{i}\frac{4\kappa+1+2k-2i}{4\kappa+1+2k-i}\frac{(-2\kappa)_{i}(1+2\kappa)_{2k-i}}{(1+4\kappa-i)_{2k}}
\\ \times
                  \frac{\Gamma(\frac{1}{2}(1+i)+\kappa)\Gamma(1+3\kappa+k-\frac{1}{2}i)}{\Gamma(\frac{1}{2}(1+i)-\kappa)\Gamma(1+\kappa+k-\frac{1}{2}i)} .
\label{N=3moment}
\end{multline}
\end{proposition}
\begin{proof}
In order to compute the moments $ m_k $ we expand the hypergeometric sum, integrate term-by-term and find an
integral of the form \eqref{Fintegral}, but with the replacements $ M \mapsto M+l/2 $ and $ k \mapsto k-M+l/2 $
for some $ l\in \Z_{\geq 0} $. Resumming this again we have \eqref{momN=3}.
This hypergeometric function with unit argument is an integer extension of a terminating Watson's Sum
(see Eq. 16.4.6 of \cite{DLMF}) for which alternative sums have recently become available - i.e. are $ 2k $-fold sums
rather than $ 2M $-fold sums. From Eq. (24) of \cite{Chu_2012} we see that what we seek is $ W_{2k,0}(a,b,c) $
with $ a=-2M, b=1+\beta, c=-\beta $ (of course $ \beta=2M $ but we only apply the termination through one parameter initially).
Thus we can utilise Theorem 5, pg. 474, of that work with the above specialisations and employing a terminating form 
of Watson's Sum we arrive at \eqref{N=3moment}.
\end{proof}

As an alternative to this one can generate the moments recursively by using the contiguous relation for the
$ {}_{3}F_{2} $ with unit argument, given by Eq. (16.4.12) of \cite{DLMF}, where we employ the abbreviation
$ F(k+1):={}_{3}F_{2}(-\beta ,-\beta ,1+\beta;-2\beta ,k+1;1) $
\begin{multline*}
   (\beta+k)^2(1+\beta-k)\left[ F(k+1)-F(k) \right]
\\
   + \beta^2(1+\beta)F(k) -k(k-1)(2+\beta-k)\left[ F(k)-F(k-1) \right] = 0 ,
\end{multline*}
with the initial values
\begin{equation*}
    F(1) = \cos\pi\kappa\, \frac{\Gamma(3\kappa+1)\Gamma(2\kappa+1)^2}{\Gamma(\kappa+1)^3\Gamma(4\kappa +1)} ,\quad
    F(2) = \frac{1}{2(2\kappa+1)}F(1) .
\end{equation*}
Finally we display some of the low moments for the purposes of checking and comparison
\begin{align}
   m_{1}+3 & = \frac{3}{1+2\kappa} ,
\label{evenN=3:1}\\
   m_{2}+3 & = \frac{6}{1+2\kappa}-\frac{6}{(1+\kappa)^2}+\frac{3}{1+\kappa} ,
\label{evenN=3:2}\\
   m_{3}+3 & = 9+\frac{9}{1+2\kappa}-\frac{36}{(1+\kappa)^2}+\frac{126}{1+\kappa}-\frac{315}{3+2\kappa} ,
\label{evenN=3:3}\\
   m_{4}+3 & = \frac{12}{1+2\kappa}-\frac{120}{(1+\kappa)^2}+\frac{654}{1+\kappa}-\frac{3780}{3+2\kappa}+\frac{360}{(2+\kappa)^2}+\frac{1254}{2+\kappa} ,
\label{evenN=3:4}\\
   m_{5}+3 & = \frac{15}{1+2\kappa}-\frac{300}{(1+\kappa)^2}+\frac{2070}{1+\kappa}-\frac{21735}{3+2\kappa}+\frac{7200}{(2+\kappa)^2}+\frac{1320}{2+\kappa}+\frac{15015}{5+2\kappa} ,
\label{evenN=3:5}\\
   m_{6}+3 & = 9+\frac{18}{1+2\kappa}-\frac{630}{(1+\kappa)^2}+\frac{5076}{1+\kappa}-\frac{85680}{3+2\kappa}
\notag\\
           & \phantom{= 9+} +\frac{62640}{(2+\kappa)^2}-\frac{104724}{2+\kappa}+\frac{450450}{5+2\kappa}-\frac{15120}{(3+\kappa)^2}-\frac{82854}{3+\kappa} .
\label{evenN=3:6}
\end{align}

\section{A brief literature survey on loop equations for circular ensembles}\label{Survey}
\setcounter{equation}{0}

In conclusion we have given a self contained and complete proof of
a hierarchy of loop equations for circular $\beta$ ensembles. We did this for the purpose
of setting up a formalism to give a systematic derivation of the sequence of
degree $k$ polynomials in the coupling $\kappa = \beta/2$ occurring as
the coefficients in the small $k$ expansion of the bulk scaled structure function
$\pi \beta S(k;\beta)/|k|$ \eqref{Sa}. To derive the loop equations, our starting point
is an adaptation of what in the theory of Selberg integrals (see e.g.~\cite{For_2010}, Ch.~4)
is known as Aomoto's method, and in particular we work directly and specifically
with the circular ensemble PDF. Our work differs from previous literature relating to
loop equations for circular ensembles in its motivation, methodology and technical
achievements, as we will indicate by giving a brief survey of some relevant literature.

As remarked in \S \ref{CircularDefn}, the loop equation formalism for circular ensembles can be traced
back to the study of (\ref{1.1a}) in the particular case $\beta = 2$ and $ V(\theta) = t\cos\theta $.
The corresponding partition function is a special case of the so called Brezin-Gross-Witten unitary
matrix model
\begin{equation}\label{ZM}
Z(M = J^\dagger J) = \int [dU] e^{- \frac{1}{g^2} {\rm Tr} \, (J^\dagger U + J U^\dagger)},
\end{equation}
where $[dU]$ is the normalised Haar measure on $ U(N) $. The work \cite{BN_1981} deduced that (\ref{ZM}),
upon the replacement $g^2 \mapsto g^2 N$ satisfies the Schwinger-Dyson equation
\begin{multline*}
   \frac{1}{g^4N^2}\sum_{i=1}x_i Z 
\\   
   = \Bigg[ \frac{1}{N}\sum_{i} x_i\frac{\partial}{\partial x_i}
                   + \frac{1}{N^2}\left( \sum_{i} x_i^2\frac{\partial^2}{\partial x_i^2}
                                            +\sum_{i\neq j}\frac{x_ix_j}{x_i-x_j}\left( \frac{\partial}{\partial x_i}-\frac{\partial}{\partial x_j}\right) \right) \Bigg] Z ,
\end{multline*}
where the eigenvalues of $J^\dagger J$ are $\{x_j\}_{j=1}^N$. With the
exponent in (\ref{ZM}) replaced by $ V = \sum_{k=-\infty}^{\infty} t_k {\rm Tr} \, U^k$,
the studies \cite{BMS_1991,GN_1992} deduced the Schwinger-Dyson equations $L_n^\pm Z = 0$, where $L_n^\pm$ are
the Virasoro operators
\begin{multline*}
   L_n^{+} = \sum_{k=-\infty}^{\infty} k t_k \left( \frac{\partial}{\partial t_{k+n}}-\frac{\partial}{\partial t_{k-n}} \right)
\\
   + \sum_{1\leq k\leq n} \left( \frac{\partial^2}{\partial t_{k}\partial t_{n-k}}+\frac{\partial^2}{\partial t_{-k}\partial t_{k-n}} \right), \quad n\geq 1,
\end{multline*}
\begin{multline*}
   L_n^{-} = \sum_{k=-\infty}^{\infty} k t_k \left( \frac{\partial}{\partial t_{k+n}}+\frac{\partial}{\partial t_{k-n}} \right)
\\
   + \sum_{1\leq k\leq n} \left( \frac{\partial^2}{\partial t_{k}\partial t_{n-k}}-\frac{\partial^2}{\partial t_{-k}\partial t_{k-n}} \right), \quad n\geq 0,
\end{multline*}

A highlight of this line of investigation, which includes \cite{PS_1990,PS_1990a,MP_1990,MP_1990a,Rob_1992}, 
was the work of Hisakado \cite{His_1996,His_1997,His_1998}.
In the latter, for the potential $V(\theta) =  t \cos \theta$, both the Toda
lattice equation and Virasoro constraints were used to characterise the corresponding partition function
in terms of a solution of the Painlev\'e III equation.

In the Introduction, we recalled some results relating to the global scaling limit of the Gaussian
$\beta$-ensembles, which for $\beta = 2$ is a particular Hermitian matrix model. In a paper
published in 2005, Mizaguchi \cite{Miz_2005} made use of the Cayley transformation 
$U = (\mathbb I + i H)/(\mathbb I - i H)$ between unitary and Hermitian matrices to initiate a study of unitary matrix integrals
with the aim of obtaining genus expansions of the free energy. By way of motivation, he writes: 
"The recent use of matrix models for
the study of gauge theory and string theory requires not only the knowledge
of their critical behaviours but also their individual higher genus corrections away
from criticality; the technology to compute them has been less developed in unitary
one-matrix models than in Hermitian ones."

In 2006 Chekhov and Eynard \cite{CE_2006} undertook a loop equation analysis of a class of $\beta$-generalised
matrix models defined by (\ref{1.1}), with $ V(x) $ analytic in $ x $ and the the eigenvalues restricted to a given
contour in the complex plane. It is also required that the absolute value signs in the product of differences
be removed. Formally at least, this includes a class of circular ensembles.
However, the correlation functions are not based on the Riesz-Herglotz kernel \eqref{R-Hkernel}, but rather the
resolvent kernel $ 1/(\zeta - z) $ familiar in the study of Hermitian matrix models. Various extensions of
this study are given in \cite{Che_2010,CEM_2010,BEMP_2012}; none treat specifically the circular ensembles nor arrange for the correlation
functions to be based on the kernel \eqref{R-Hkernel}. Analytic features particular to circular ensembles (or more generally closed
contours), such as the need
to consider the domain inside, and the domain outside, the unit circle on equal footing do not show themselves. 

Thus, by treating the circular ensemble directly, we have been able to make stronger analytic statements than hold for
a $\beta$ ensemble on a general curve.  By way of application, we have been able to provide a computational scheme for the problem
at hand, namely the systematic derivation of the polynomials in \eqref{Sa}, and this in turn has
lead to the discovery of some new rational function structures for the moments $m_k$ in expansion \eqref{mDefn} 
as given in Proposition \ref{mRational}.

\section*{Acknowledgements}
The work of NSW and PJF was supported by the Australian Research Council Discovery Project DP110102317.
Correspondence with B. Eynard, initiating \S \ref{Survey}, is acknowledged.

\bibliographystyle{plain}
\bibliography{moment,random_matrices,nonlinear,CA}

\end{document}